\documentclass[a4paper,11pt]{article}

\usepackage[english]{babel}
\usepackage[a4paper]{geometry}
\usepackage{enumerate,enumitem}
\usepackage{amsmath}
\usepackage{dsfont}
\usepackage{amsthm}
\usepackage{amssymb}
\usepackage{hyperref}
\usepackage{amssymb}
\usepackage{color}
\usepackage[utf8]{inputenc}
\usepackage{graphicx}
\usepackage{tikz}
\usepackage{tikz-cd}
\usepackage{verbatim} 
	\usepackage{mathtools} 
	\usepackage{mathabx} 
	\usepackage{hyperref}
	\usepackage{comment}
	\usetikzlibrary{decorations.markings}
	\usepackage{bbm} 
	\usepackage{xfrac} 

	\usepackage{pgfplots}              
	\pgfplotsset{compat=newest}
	\usepgfplotslibrary{fillbetween}
	
	\newtheorem{theorem}{Theorem}[section]  
	\newtheorem{cor}[theorem]{Corollary}
	\newtheorem{prop}[theorem]{Proposition}
	\newtheorem{lemma}[theorem]{Lemma}

	\numberwithin{equation}{section}
	\theoremstyle{definition}
	
	\newtheorem{rem}[theorem]{Remark}

	\newcommand{\caF}{{\mathcal F}}
	
	\newcommand{\caH}{{\mathcal H}}

	\newcommand{\caK}{{\mathcal K}}
	\newcommand{\caL}{{\mathcal L}}
	
	\newcommand{\caN}{{\mathcal N}}

	\newcommand{\caQ}{{\mathcal Q}}

	\newcommand{\caV}{{\mathcal V}}

	\newcommand{\bbR}{{\mathbb R}}

\begin{document}
		\title{Two-sided Lieb-Thirring bounds}
		
		\author{Sven Bachmann$^1$, Richard Froese$^1$, Severin Schraven$^1$ \\
			\\
			$^1$ Department of Mathematics, The University of British Columbia, \\
			Vancouver, BC V6T 1Z2, Canada}
		
		\maketitle
		
		\begin{abstract}
We prove upper and lower bounds for the number of eigenvalues of semi-bounded Schr\"odinger operators in all spatial dimensions. As a corollary, we obtain two-sided estimates for the sum of the negative eigenvalues of atomic Hamiltonians with Kato potentials. Instead of being in terms of the potential itself, as in the usual Lieb-Thirring result, the bounds are in terms of the landscape function, also known as the torsion function, which is a solution of $(-\Delta + V +M)u_M =1$ in $\bbR^d$; here $M\in\bbR$ is chosen so that the operator is positive. We further prove that the infimum of $(u_M^{-1} - M)$ is a lower bound for the ground state energy $E_0$ and derive a simple iteration scheme converging to $E_0$. 
		\end{abstract}

		\section{Introduction}\label{sect:intro}

The Lieb-Thirring (LT) inequalities
\begin{equation*}
\mathrm{tr}((-\Delta +V)_-^\gamma) \leq L_{\gamma,d} \int_{\mathbb{R}^d} V_-(x)^{\gamma+\frac{d}{2}} dx
\end{equation*}
are bounds on the moments of the negative eigenvalues of atomic Hamiltonians in terms of the negative part of the potential~$V=V_+ - V_-$. In the context of the stability of matter for which it was originally developed~\cite{lieb1976inequalities}, only an upper bound was needed and much of the vast amount of subsequent work extending and refining the estimates, see for example~\cite{sh1979asymptotic, frank2022schrodinger}, has similarly focussed on upper bounds. Among the few results about lower bounds, we mention~\cite{schmincke1978schrodinger, grigor2004eigenvalues, damanik2007schrodinger} in special cases as well as~\cite{safronov2008estimates, molchanov2012bargmann} for somewhat different estimates.

In this work, we give a simple proof of a two-sided estimate on the number of eigenvalues and their moments. The upper and lower bounds are matching in the sense that they differ only by a scaling and a multiplicative constant, but we are not able to conjecture the sharp values. The main difference with the LT inequality is that the bounds are not in terms of the potential $V$, but rather in terms of what we shall refer to as the \emph{effective potential} $\frac{1}{u_M}$ where $u_M$ is a solution of 
\begin{equation}\label{landscape1}
(-\Delta + V +M\cdot\mathds{1})u_M =1.
\end{equation}
Unlike in~\cite{bachmann2023counting} where the potential $V$ is assumed (among other hypotheses) to be a nonnegative function, our assumption here is that the corresponding Schr\"odinger operator is semi-bounded below,
\begin{equation*}
E_0 = \inf\mathrm{spec}(-\Delta + V) >-\infty.
\end{equation*}
The constant in~(\ref{landscape1}) is then $M>-E_0$. Among the auxiliary results of this paper, we will in particular show the existence of a weak solution $u_M$ in the whole space, see also~\cite{poggi2021applications}.

The equation \eqref{landscape1} has been extensively studied without the potential term, i.e. for $-\Delta u_M=1$, and its solution is known as the \emph{torsion function}. It has been used to establish optimal Sobolev constant, estimate the ground state energy of the Laplacian and to obtain pointwise estimates for the eigenfunctions (and more general solutions) of elliptic operators \cite{payne1981bounds, sperb2003bounds, van2009hardy, giorgi2010principal, van2012estimates, van2014torsion, aghajani2014pointwise, van2017spectral}. Following this line of work, the more general case with a potential term was considered only recently \cite{vogt2019estimates}.

The connection of the torsion function with the phenomenon of localization for random Schr\"odinger operators was discovered in~\cite{filoche2012universal} in the case $M=0$. In this context, it is mostly called the \emph{landscape function} and it has been observed to yield excellent estimates on the density of states throughout the spectrum~\cite{arnold2019computing, david2021landscape, desforges2021sharp, arnold2022landscape} and to predict both the localization centres and exponential decay of localized eigenfunctions~\cite{arnold2019localization, filoche2021effective, wang2021exponential}, see also~\cite{berg2022efficiency}. As pointed out in \cite{mugnolo2023pointwise} (or also implicitly in \cite[page 2903]{steinerberger2017localization}) the predicting power of the landscape function for the localization of eigenfunctions relies solely on the fact that the resolvent of the Schr\"odinger operator is positivity-preserving, a key property that we will use here as well.

The present work and our previous~\cite{bachmann2023counting} focus on spectral estimates using the landscape function in the non-random case. In this context, the effective potential plays the role of a smoothing of the original potential at the right scale: It allows in particular for Cwikel-Lieb-Rozenblum (CLR)-type bounds \cite{lieb1976bounds, rozenblum1976distribution, cwikel1977weak} in cases where a strict CLR bound in terms of $V$ fails. Specifically, the eigenvalue counting function can be estimated in terms of the volume of the sublevel sets of the effective potential. In this context, we mention~\cite{poggi2021applications,hoskins2024magnetic} for a magnetic version of the landscape function as well as the closely related approaches via Brownian motion \cite{steinerberger2017localization}, regularizing kernel \cite{lu2022fast}, respectively maximal functions \cite{otelbaev1976bounds, otelbaev1979imbedding, shen1996eigenvalue, shen1998bounds, bugliaro1997lieb, shen1998moments, do2023spectral}. 

We shall see here that the key property for the effective potential to be useful in estimating the eigenvalues of the Schr\"odinger operator is that $u_M$ satisfies a Harnack inequality, see~(\ref{Harnack Moser inequality}) below. While we are not aware of the most general assumptions under which this holds, the results of~\cite{AizenmanSimon, chiarenza1986harnack} yield that it does for atomic potentials, namely for negative $V$'s that are also in the Kato class~(\ref{KatoClass}), which are our main motivation here.

We state our main results in details in Section~\ref{sect:results}. In Section~\ref{sect:CLR} we prove an abstract two-sided CLR-type bound conditional on the existence of a landscape function in the whole space and on a Harnack inequality. We then prove the validity of these assumptions in the case of Kato-class potentials in Section~\ref{sect:existence}. In this case, the bounds can in fact be notably improved. In Section~\ref{sect:groundstate}, we concentrate on the relationship between the infimum of the effective potential and the bottom of the spectrum: In particular, we propose an iterative scheme which we converges to the ground state energy. Section~\ref{sect:examples} serves as both illustration of the results and analysis of their sharpness: On the one hand, in the case of the shallow square well in one dimension, we show how well the bottom of the effective potential approximates the ground state energy, and on the other hand we prove that our bounds match the asymptotic distribution of the eigenvalues at the bottom of the essential spectrum, in the examples of attractive radial potentials with a power law singularities. The proof of the two-sided Lieb-Thirring inequalities in terms of $\frac{1}{u_M}-M$ is provided in Section~\ref{sect:LT}. 
		
		\section{Results}\label{sect:results}
	
One of the central objects of interest in this paper is the eigenvalue counting function for the Schr\"odinger operator $-\Delta + V$ on $L^2(\bbR^d)$, which is the rank of the spectral projection associated with $(-\infty,\mu]$, namely
\begin{equation}\label{def:counting}
	\mathcal{N}^{V}(\mu) = \mathrm{dim}\, \mathrm{Ran}\left(\mathbbm{1}_{(-\infty,\mu]}(-\Delta +V)\right).
\end{equation}
While most of the existing results in the literature are bounds on $\mathcal{N}^{V}$ in terms of certain integrals of $V$, we aim here for estimates in terms of the effective potential, $\frac{1}{u_M}$ where $u_M$ is, formally, the solution of~(\ref{landscape1}). To start with, we shall assume its existence, namely that there is a $u_M\in H_\mathrm{loc}^1(\mathbb{R}^d)$ such that 
\begin{equation} \label{eq:weak form}
	\int_{\mathbb{R}^d} \big(\nabla u_M(y) \cdot \nabla \varphi(y) + (V(y)+M) u_M(y) \varphi(y) \big) dy = \int_{\mathbb{R}^d} \varphi(y) dy
\end{equation}
for all $\varphi\in C_c^\infty(\mathbb{R}^d) $. Our first bounds will be in terms of the volume~$\mathcal{V}_M(\mu)$ of the sublevel sets of the effective potential
\begin{equation*}
\mathcal{V}_M(\mu) = \int_{\{x\in \mathbb{R}^d \ : \ \frac{1}{u_M}\leq \mu \}} dx.
\end{equation*}
For the bound to hold, we assume what we shall refer to as a Harnack-Moser inequality, namely that there is $C_{HM}>0$ such that 
\begin{equation} \label{Harnack Moser inequality}
		\sup_{Q} u_M \leq C_{HM} \left( \inf_Q u_M + \ell(Q)^2 \right)
\end{equation}
for all cubes $Q\subseteq \mathbb{R}^d$ with sidelength $\ell(Q)$ (for exact definitions, see Section~\ref{sect:CLR}).

\begin{theorem} \label{thm:main}
	Let $V\in L^1_\mathrm{loc}(\mathbb{R}^d)$ be such that
\begin{equation*}
E_0=\inf \mathrm{spec}(-\Delta +V)>-\infty.
\end{equation*}
Let $M>-E_0$ and assume that there is a solution $u_M$  of~\eqref{eq:weak form} which is positive, such that $u_M\in H_\mathrm{loc}^1(\mathbb{R}^d)\cap L^\infty_\mathrm{loc}(\mathbb{R}^d)$ and $\frac{1}{u_M}\in L^\infty_\mathrm{loc}(\mathbb{R}^d)$. If additionally~(\ref{Harnack Moser inequality}) holds, then there exist constants $c_{0,d},C_{0,d}>0$ depending only on $d, C_{HM}$ such that
\begin{equation} \label{CLR}
		(c_{0,d}\mu)^{d/2} \mathcal{V}_M(c_{0,d}\mu) \leq \mathcal{N}^{V+M}(\mu) \leq  (C_{0,d}\mu )^{d/2} \mathcal{V}_M(C_{0,d}\mu),
\end{equation}
 for all $\mu \in \mathbb{R}$. 

		\end{theorem}
\noindent Note that the constants $c_{0,d}, C_{0,d}$ depend on $M$ only via $C_{HM}$. 
\begin{rem}
	\begin{enumerate}
		\item The assumptions on the solution $u_M$ should be understood as implicit assumptions on the potential $V$. There are two large classes for which they can be proved to hold: $(i)$ When $V$ is Kato-class (see Theorem \ref{thm:Kato} below), which is the case for the physically relevant case of atomic potentials, $(ii)$ When $V$ is a potential that is bounded below and satisfies a scale-invariant Kato inequality (see \cite{bachmann2023counting}); this includes in particular all polynomials which are bounded from below. In both cases it is key to have suitable kernel estimates for the corresponding Green's function. For non-negative potentials the existence of weak solutions in $W_\mathrm{loc}^{1,2}(\mathbb{R}^d)\cap L_\mathrm{loc}^2(\mathbb{R}^d)$ of \eqref{eq:weak form} (without requiring \eqref{Harnack Moser inequality}) has been studied extensively in \cite{poggi2021applications}.
		\item We do not claim that the solution satisfying the requirement of Theorem~\ref{thm:main} is unique. While we are not aware of a potential admitting multiple such solutions, the conclusion would hold for any solution of this type. Consider for example $V=0$ and $d=1$. In this case any $M>0$ is admissible. All solutions of \eqref{eq:weak form} are of the form $M^{-1}+c_1 e^x+c_2 e^{-x}$, however, only the constant solution satisfies \eqref{Harnack Moser inequality}. \newline
		Both for Kato class potentials and the potentials considered in \cite{bachmann2023counting}, one can show that 
		\begin{equation} \label{uM as integral of Greens function}
			u_M(x) = \int_{\mathbb{R}^d} G_{M}(x,y) dy
		\end{equation}
		is an admissible solution where $G_M$ is the Green's function associated to $-\Delta+V+M$ (for the constant potential just discussed, the explicit form of the heat kernel immediately yields that the constant solution coincides with the function given by \eqref{uM as integral of Greens function}). For $V\geq 0$ the solution in \eqref{uM as integral of Greens function} can be realized as pointwise limit of solutions of the landscape equation on balls of increasing size with Dirichlet boundary conditions (see \cite[Theorem 1.18]{poggi2021applications}). Both for Kato-class potentials and in the setting of \cite{bachmann2023counting}, we prefer to realize an admissible solution as pointwise limit of solutions where the constant function on the RHS of \eqref{eq:weak form} is truncated, namely $u_M(x) = \lim_{R\rightarrow \infty} (-\Delta+V+M)^{-1} \mathbbm{1}_{B(0,R)}$.		
		\newline
		Another example where the solution of the landscape equation can be handled completely is that of the Coulomb potential in $d=3$. Since $-\Delta-\vert x \vert^{-1}+M$ commutes with rotation, the function $(-\Delta-\vert x\vert^{-1}+M)^{-1} \mathbbm{1}_{B(0,R)}$ is radially symmetric and so is $u_M$ in \eqref{uM as integral of Greens function}. In this case, the two linearly independent radial solutions of the homogeneous equations have a singularity, at zero and infinity respectively, and the landscape equation has a unique bounded and radial solution.
		\item While Theorem \ref{thm:main} is reminiscent of the landscape law established in \cite{david2021landscape}, there are three essential differences which are all not only natural but also desirable for the relevant quantum mechanical applications. First of all, we are working in infinite volume, where even the existence of the landscape function is a priori unclear. Secondly, we work under the more general assumption that the operator is bounded below rather than the potential being bounded below. Finally, our estimates are in terms of the measure of the sublevel sets rather than a box-counting measure. The variational argument of Section \ref{sect:CLR} are inspired by the one in \cite{david2021landscape} although we need to modify it to accommodate potentials with a possibly singular negative part. Furthermore, \cite{david2021landscape} requires the absolute value squared of the landscape function to be doubling, rather than the Moser-Harnack inequality \eqref{Harnack Moser inequality}. This weaker assumption is sufficient there as the positivity of the potential implies $-\Delta u\leq 1$, which in turn yields a subsolution estimate, namely \eqref{Harnack Moser inequality} with $\inf u$ replaced by an $L^2$ average. One could weaken \eqref{Harnack Moser inequality} to the averaged version of \cite{david2021landscape} (yielding the lower bond on $\caN^{V+M}$) and a suitable weighted Poincaré inequality (for the upper bound). \newline
		While the full power of \eqref{Harnack Moser inequality} is indeed not required to obtain two-sided estimates of $\caN^{V+M}$ in terms of the coarse-grained volumes, it is important for us as it allows to pass from the coarse-grained volume to the true measure of the sublevel set. On the one hand, this implies the Lieb-Thirring bounds. On the other hand, the fact that no coarse graining is needed emphasises the role of the transformation $V\mapsto \frac{1}{u_M}$ as a smoothing at the relevant scale. In fact, the validity of estimates in terms of the measure of the sublevel sets was conjectured in~\cite{david2021landscape}.
		\newline
		Finally, we point out that the existence of a landscape function for nonnegative potentials has recently been established in a probabilistic setting \cite{david2023landscape}. 
	\end{enumerate}
\end{rem}
		
		This bound in terms of the volume of the sublevel sets implies the following Lieb-Thirring bounds (see \cite{read2023lieb} for a study on how shifting affects the Lieb-Thirring inequalities in a similar setting).
		
\begin{cor} \label{cor:LiebThirring}
Under the assumptions of Theorem~\ref{thm:main} and if $E_0<0$, then
\begin{equation} \label{Lieb-Thirring lower}
	 c_{\gamma,d} \int_{\mathbb{R}^d} \left(\vert E_0\vert+\delta-\frac{1}{c_{0,d}u_M(x)}\right)_{+}^{\gamma+\frac{d}{2}} dx
	 \leq \mathrm{tr}((-\Delta+V)_-^\gamma)
\end{equation}
for all $\gamma>0$, and  if $\gamma\geq 1$,
\begin{equation} \label{Lieb-Thirring upper}
		\mathrm{tr}((-\Delta+V)_-^\gamma) \leq C_{\gamma, d}  \int_{\mathbb{R}^d} \left(\vert E_0\vert +2\delta-\frac{1}{C_{0,d}u_M(x)}\right)_{+}^{\gamma+\frac{d}{2}} dx.
\end{equation}
Here, $\delta =M-\vert E_0\vert$ and
\begin{equation*}
c_{\gamma,d} = \frac{c_{0,d}^{d/2}\gamma}{\frac{d}{2}+\gamma} \min\left\{1, \frac{\delta}{\vert E_0\vert}\right\}^\frac{d}{2} ,\qquad 
C_{\gamma,d} = \frac{C_{0,d}^\frac{d}{2}\gamma}{\frac{d}{2}+\gamma} \left(1+\frac{\vert E_0\vert}{\delta}\right)^\frac{d}{2} ,
\end{equation*}
where $c_{0,d},C_{0,d}$ are the constants of Theorem~\ref{thm:main} for the potential $V+M$. 
\end{cor}

The Harnack-Moser inequality \eqref{Harnack Moser inequality} controls the oscillations of $u_M$ at all scales and is the natural choice in the present context. It originates in the PDE literature in slightly different settings. On finite domains with $V\equiv 0$ this inequality holds always true for nonnegative solutions of \eqref{eq:weak form}, see \cite[Theorem 3.3]{han2011elliptic}). On the other hand, for Kato-class potentials (see \cite[Theorem 2.5]{chiarenza1986harnack}) it is known that nonnegative solutions of the homogeneous equation $(-\Delta +V)u=0$ satisfy a Harnack inequality for balls with small radius, namely $\sup_{B(x,r)} u \leq C_H \inf_{B(x,r)} u$. The key difficulty to overcome here is the lack of compactness as we are working in $\mathbb{R}^d$ and need the inequality to hold for arbitrary large cubes.
	
The assumptions of the above theorem, and in particular~\eqref{Harnack Moser inequality} can be verified for the physically relevant potentials in the Kato-class $\caK_d$ (see Section~\ref{sect:existence} for a precise definition). This includes all bounded potentials as well as power laws singularities $\vert x \vert^{-\rho}$ for $\rho \in [0,\min \{d,2\})$. A `crystal' of Coulomb singularities is allowed too.
		
\begin{prop} \label{prop:existence}
		Let $V$ be such that $V_+\in \mathcal{K}_{d}^\mathrm{loc}, V_-\in \mathcal{K}_d$. Then $E_0= \inf \mathrm{spec}(-\Delta +V)>-\infty$. For all $M>- E_0$ there exists a positive solution $u_M$ of \eqref{eq:weak form} such that $u_M\in L^\infty(\mathbb{R}^d) \cap C^0(\mathbb{R}^d)\cap H_\mathrm{loc}^1(\mathbb{R}^d)$ and $\frac{1}{u_M} \in L^\infty_\mathrm{loc}(\mathbb{R}^d)$. If, additionally, $V_+\in \mathcal{K}_d$, then $u_M^{-1}\in L^\infty(\mathbb{R}^d)$ and $u_M$ satisfies \eqref{Harnack Moser inequality}. In fact, 
\begin{equation*}
0<A_M=\frac{\sup_{\mathbb{R}^d} u_M}{\inf_{\mathbb{R}^d} u_M}<\infty.
\end{equation*}
\end{prop}
\noindent The fact that, with a Kato condition, the operator $-\Delta + V$ is bounded from below is well-known, see~\cite[Theorem A.2.7.]{simon1982schrodinger}). Note that the very last claim of the theorem is a global Harnack inequality which implies~\eqref{Harnack Moser inequality}, although for a potentially suboptimal constant since $C_{HM}\leq A_M$.

The bound~(\ref{CLR}), while very general indeed, is not optimal as the spectral parameter is shifted by $M$. In other words, a sharp estimate would use $\frac{1}{u_M} - M$ instead of $\frac{1}{u_M}$ as the effective potential. The following theorem realizes this improvement for Kato-class potentials. Specifically, we consider the unshifted counting function $\mathcal{N}^V(\mu)$ and show (i)~that a lower bound in terms of coarsed-grained volume of the sublevel sets of $\frac{1}{u_M} - M$ always holds, (ii)~that the coarse-graining can be removed if $\frac{1}{u_M} - M$ satisfies a scale-invariant Harnack inequality and (iii)~that a CLR-type upper bound, with $V$ replaced by $\frac{1}{u_M} - M$, holds in dimensions $3$ and higher.
		
		\begin{theorem} \label{thm:Kato}
			Let $V\in \mathcal{K}_d$ and $M>- E_0$. Let $u_M$ be the function given by Proposition~\ref{prop:existence}.
			\begin{enumerate}[label=\roman*)]
				\item Let $\mu <0$. For any $c>1$, let
				\begin{equation} \label{def tilde n}
					n_c(\mu) = \left\vert \left\{ Q\in \mathcal{Q}_{(C_c\vert \mu\vert)^{-1/2}} \ : \ \sup_Q \left(\frac{1}{u_M} - M \right) \leq c\mu  \right\} \right\vert,\quad C_c = \frac{c-1}{2^d(5 A_M)^2},
				\end{equation}
where $\caQ_\ell$ is a partition of $\mathbb{R}^d$ into cubes of side length $\ell$, see Section \ref{sect:CLR}. Then
				\begin{equation}
					\mathcal{N}^V(\mu)\geq n_c(\mu).
				\end{equation}				
				\item Let $c>1$. Assume there exists $\widetilde{C}_H>0$ such that for all $\ell>0$ and all $Q\in \mathcal{Q}_{\ell}$ with $\sup_Q (\frac{1}{u_M}-M)\leq -\frac{c}{C_c\ell^2}$,
				\begin{equation} \label{scale invariant Harnack}
					\sup_Q \left( \frac{1}{u_M} -M \right) \leq \widetilde C_H\inf_Q \left(\frac{1}{u_M}-M\right).
				\end{equation} 
				Then
				\begin{equation}\label{Lower improved Kato}
					\mathcal{N}^V(\mu)\geq (C_c\vert \mu\vert)^{d/2} \left\vert \left\{ x\in \mathbb{R}^d \ : \ \frac{1}{u_M(x)}-M \leq \frac{c\mu}{\widetilde{C}_H} \right\} \right\vert
				\end{equation}
				 for all $\mu < 0$.
				
				\item Let $d\geq 3$. There exist a constant $C_\mathrm{CLR}>0$ depending only on $d$ such that
				\begin{equation} \label{CLR improved}
					\mathcal{N}^V(\mu) \leq C_\mathrm{CLR} A_M^d \int_{\mathbb{R}^d} \left( (\mu+\varepsilon)- \left(\frac{1}{u_M(x)}-M\right)\right)_+^{d/2} dx,
				\end{equation}
				 for all $\mu\in \mathbb{R}$ and all $\varepsilon>0$.
			\end{enumerate} 
		\end{theorem}
		
Here again, these bounds yield upper and lower Lieb-Thirring inequalities. The difference with the bounds of Corollary~\ref{cor:LiebThirring} cannot be understated: Here, the effective potential appears as it really should, without an additional scaling and shift. 

\begin{cor} \label{cor:Kato LiebThirring}
Under the assumptions of Theorem~\ref{thm:Kato}(ii), respectively~(iii),
\begin{equation} \label{Kato Lieb-Thirring lower}
	 k_{\gamma,d}\int_{\mathbb{R}^d} \left(\frac{1}{u_M(x)}-M\right)_{-}^{\gamma+\frac{d}{2}} dx
	 \leq \mathrm{tr}((-\Delta+V)_-^\gamma)
\end{equation}
and
\begin{equation} \label{Kato Lieb-Thirring upper}
		\mathrm{tr}((-\Delta+V)_-^\gamma) \leq L_{\gamma, d} A_M^{2\gamma +d}  \int_{\mathbb{R}^d} \left(\frac{1}{u_M(x)}-M\right)_{-}^{\gamma+\frac{d}{2}} dx,
\end{equation}
for all $\gamma>0$ where $L_{\gamma,d}$ are the standard Lieb-Thirring constants and $$k_{\gamma,d} = C_c^\frac{d}{2}\frac{\gamma}{\gamma + \frac{d}{2}}\left(\frac{\widetilde C_H}{c}\right)^{\gamma + \frac{d}{2}}.$$
\end{cor}
Using the properties of the resolvent (see \cite[Theorem B.2.1.]{simon1982schrodinger}) one knows that if $V\in L^{d/2+\gamma}(\mathbb{R}^d)$, then $u_M^{-1}-M \in L^q(\mathbb{R}^d)$ for all $q\in [d/2+\gamma,\infty)$. Thus, the upper bound in \eqref{Kato Lieb-Thirring upper} remains finite for all $\gamma>0$ in that setting, which is in general not the case for the standard Lieb-Thirring bound.

\subsection{Applications \& Examples}

While the above results are valid throughout the spectrum, their sharpness comes into focus at two critical points: the bottom of the spectrum and the bottom of the essential spectrum.

\subsubsection{On the ground state energy}

The first proposition below shows that the infimum of $\frac{1}{u_M} - M$ is a lower bound for the ground state energy $E_0$, see also~\cite{berg2022efficiency} for nonnegative potentials as well as the previous work \cite{van2009hardy}. This is turn allows us to set up an iterative procedure converging to $E_0$. Finally, we show that the CLR bound using $\frac{1}{u_M} - M$ yields the exact asymptotics of the negative eigenvalues accumulating at $0^-$ for `atomic' potentials $V(x)=-\vert x \vert^{-\rho}$.

		\begin{prop} \label{prop:groundstate}
			Let $V\in \mathcal{K}_d$ and $M>- E_0$. Let $u_M$ be the function given by Proposition~\ref{prop:existence}. We have
			\begin{equation} \label{eq:ground state}
				 E_0 \geq \inf_{\mathbb{R}^d} \left(\frac{1}{u_M}-M\right).
			\end{equation}
		\end{prop}
It may seem that the proposition is relatively useless without a priori information on the ground state energy. This is erroneous as the result can be bootstrapped to provide a convergent sequence converging to $E_0$. As we shall see in Section~\ref{sect:examples}, the approximation may already be very good even after one iteration, even if $M-E_0$ is very large.
\begin{cor} \label{cor:iteration}
	Let $V\in \mathcal{K}_d$ and let
		\begin{align*}
			F: (-E_0, \infty) \rightarrow (-E_0,\infty), \quad M\mapsto M-\inf_{\mathbb{R}^d} \frac{1}{u_M},
		\end{align*}
	Let $M_0\in (-E_0 , \infty)$ and define recursively $M_{n+1}=F(M_n)$ for $n\in\mathbb{N}$. Then $M_{n+1}\leq M_n$ and $M_n \to -E_0 $ as $n\rightarrow \infty$. Moreover,
			\begin{equation} \label{eq:blowup}
				\lim_{M\rightarrow (-E_0)^+} \Vert u_M \Vert_{L^\infty(\mathbb{R}^d)} =\infty.
			\end{equation}
		\end{cor}
		We remark that \eqref{eq:blowup} is in stark contrast to the random case. A shifted landscape function was considered in~\cite{david2023landscape} in the case of a random potential. There, it is shown that the limit $M\to -E_0^+$ can be taken after averaging over the random potential. In the present case where no averaging is available, we will see explicitly in Section~\ref{sect:examples} in the case of a radial potential well that $u_M$ blows up at every single point in the same limit, namely $\lim_{M\rightarrow (-E_0)^+} (u_M^{-1}-M)=E_0$. Thus, in the deterministic case, the only spectral information surviving the limiting procedure is the ground state energy.

Finally, we point out that our proofs allow us to treat magnetic Schrödinger operators with suitably well-behaved magnetic fields. 
	
\subsubsection{Asymptotics at the bottom of the essential spectrum: The atomic case}

We have just discussed how the landscape function provides a sharp estimate on the ground state energy. Similarly, we now show that the landscape functions yields the sharp eigenvalue asymptotics below the essential spectrum in the case of radial inverse power law potentials.
		
		\begin{prop} \label{prop:powerlaw}
			Let $\rho_d =\min \{d, 2\}$. For $\rho \in (0,\rho_d)$ the potential $V(x)=-\vert x \vert^{-\rho}$ is in the Kato class $\caK_d$. Let $M>- E_0$ and let $u_M$ be the function given by Proposition~\ref{prop:existence}. Then as $\mu \rightarrow 0^-$
		\begin{align}
			\caN^{V}(\mu) &= (1+o(1)) \left[ (2\sqrt{\pi})^d \Gamma\left( \frac{d}{2}+1 \right) \right]^{-1} \int_{\mathbb{R}^d} \left( \mu - \left(\frac{1}{u_M(x)}-M\right) \right)_+^{d/2} dx \label{CLR asymptotics 1}\\
			 &= (1+o(1)) \caN^{\frac{1}{u_M}-M}(\mu).\label{CLR asymptotics 2}
		\end{align}
	\end{prop}
The first equality above shows the exactness of the CLR asymptotics when using the landscape function $\frac{1}{u_M}-M$, with the same constant as with the original potential $V$. In the second equality, we point out that $\caN^{\frac{1}{u_M}-M}$ is the counting function for the operator $-\Delta + (\frac{1}{u_M}-M)$: This is non-trivial since $-\Delta + V$ is unitarily equivalent to $-\frac{1}{u_M^2}\nabla\cdot u_M^2 \nabla+ (\frac{1}{u_M}-M)$, namely for a modified kinetic energy. This proves a conjecture of~\cite[Equation $(1.5)$]{david2021landscape} in this particular case.

		\section{Variational argument: Proof for Theorem \ref{thm:main}} \label{sect:CLR}
		
We follow the strategy of \cite{bachmann2023counting} and first estimate $\mathcal{N}^V$ in terms of some coarse-grained volume of the sublevel sets of the effective potential. In a second step we will relate this to the measure of the sublevel set of $1/u_M$.
		
		We first introduce some notation. A box of sidelength $\ell$ is a set of the form $\times_{i=1}^d[a_i,b_i]$ where $b_i-a_i = \ell$ and $a_i, b_i \in \ell \mathbb{Z}$. For any $\ell>0$, we consider the collection $\caQ_\ell$ of boxes of sidelength $\ell$ such that $\bigcup_{Q\in\caQ_\ell}Q = \bbR^d$ and $\mathring Q \cap \mathring Q' = \emptyset$ whenever $Q\neq Q'$. We define for any $\mu>0$
		\begin{align*}
			N(\mu) = \left\vert \left\{ Q\in \caQ_{\mu^{-1/2}} \ : \ \inf_Q \frac{1}{u_M} \leq \mu  \right\}\right\vert
		\end{align*}
		and
		\begin{align*}
			n(\mu) = \left\vert \left\{ Q\in \caQ_{\mu^{-1/2}} \ : \ \sup_Q \, \frac{1}{u_M} \leq \mu  \right\}\right\vert.
		\end{align*}
		The reason we are considering only cubes with corner on $\ell \mathbb{Z}^d$ is to make it possible to compare $\caQ_{\ell}$ with $\caQ_{n\ell}$ for $n\in \mathbb{N}$, the former being a refinement of the later.
		
		For the class of potentials considered here, namely those satisfying the Harnack condition~(\ref{Harnack Moser inequality}), both coarse-grained volumes are directly related to the measure $\caV(\mu)$ of the sublevel set. The following lemma is essentially the same as \cite[Lemma 4.1]{bachmann2023counting}.
		
		\begin{lemma}  \label{lm:comparing}
			Let $V$ satisfy the conditions of Theorem~\ref{thm:main}. Then
			\begin{align*} 
				n(\mu) \leq \mu^{d/2} \caV_M(\mu) \leq N(\mu) \leq n\Big(\left\lceil \sqrt{2 C_{HM}} \right\rceil^2 \mu\Big)
			\end{align*}
			for all $\mu \in \mathbb{R}$. Here $C_{HM}$ is the Harnack-Moser constant of~\eqref{Harnack Moser inequality}.
		\end{lemma}
		\begin{proof}
			The first two inequalities are immediate as, up to null sets, $n(\mu)/\mu^{d/2}$ is the measure of all boxes that are strictly contained in the sublevel set $\{1/u_M \leq \mu \}$ and $N(\mu)/\mu^{d/2}$ is the measure of all the boxes that intersect the sublevel set.
			
			Let $K=\lceil \sqrt{2C_{HM}} \rceil^2$, then the cubes in $\caQ_{(K\mu)^{-1/2}}$ are subcubes of exactly one cube in $\caQ_{\mu^{-1/2}}$ since $\sqrt{K}$ is an integer. To prove the last inequality, it is sufficient to show that every cube in $\caQ_{\mu^{-1/2}}$ which contributes to $N(\mu)$ admits a subcube in $\caQ_{(K\mu)^{-1/2}}$ which contributes to $n(K \mu)$.	If $Q\in \caQ_{\mu^{-1/2}}$ satisfies $\sup_Q \frac{1}{u_M} \leq \mu$ then all its subcubes contribute to $n(K \mu)$. The only other option for $Q$ to contribute to $N(\mu)$ is if $\inf_Q \frac{1}{u_M} \leq \mu < \sup_Q \frac{1}{u_M}$. For those, we have in particular $\sup_Q u_M \geq \frac{1}{\mu}$. Now pick a subcube $\widetilde{Q}\in \caQ_{(K\mu)^{-1/2}}$ of $Q$ such that $\sup_{\widetilde{Q}} u_M \geq \frac{1}{\mu}$. Then by \eqref{Harnack Moser inequality} we get, as $K=\lceil \sqrt{2C_{HM}} \rceil^2 \geq 2C_{HM}$,
			
			\begin{align*}
				\inf_{\widetilde{Q}} u_M \geq \frac{1}{C_{HM}} \sup_{\widetilde{Q}} u_M - \frac{1}{K \mu} \geq \frac{1}{C_{HM} \mu} - \frac{1}{2 C_{HM} \mu} = \frac{1}{2C_{HM} \mu}.
			\end{align*} 
			Thus, $\sup_{\widetilde{Q}} \frac{1}{u_M} \leq 2C_{HM} \mu \leq K \mu$, and so $\widetilde{Q}$ contributes to $n(K \mu)$.
			\end{proof}

We now structure the proof of Theorem~\ref{thm:main} in two lemmas for the upper, respectively the lower bound. 
		
		\begin{lemma}
			\label{lma:UpperBound}
			Let $V$ satisfy the conditions of Theorem~\ref{thm:main}. Then
			\begin{align*}
				\mathcal{N}^{V+M}(\mu) \leq N(C\mu)
			\end{align*}
			for all $C>\frac{4d C_H^2}{\pi^2}$ and all $\mu\in \mathbb{R}$.
		\end{lemma}
		\begin{proof}
			For $\mu\leq 0$ we have $\mathcal{N}^{V+M}(\mu) =0= N(C\mu)$. Thus, we will assume for the rest of the proof that $\mu>0$. In order to have that $\mathcal{N}^V(\mu) \leq N$ it suffices, by the Min-Max Principle (see \cite[Theorem XIII.2]{reed1978iv}), to exhibit a subspace $\mathcal{H}_N\subseteq \mathrm{dom}(H^{1/2})$ with codimension at most $N$ such that
			\begin{align*}
				\int_{\mathbb{R}^d} \left( \vert \nabla v \vert^2 + (V+M) \vert v \vert^2 \right) > \mu \int_{\mathbb{R}^d} \vert v \vert^2
			\end{align*}
			for all $v \in \mathcal{H}_N$. Let $\mathcal{F}$ be the collection of boxes such that 
			\begin{align*}
				\mathcal{F} = \left\{ Q\in \caQ_{(C\mu)^{-1/2}} \ : \ \inf_Q \, \frac{1}{u_M} \leq \mu \right\},
			\end{align*}
			where $C>0$ will be chosen later, and let
			\begin{align*}
				\mathcal{H}_N = \left\{ v \in \mathrm{dom}(H^{1/2}) \ : \  \int_Q \frac{v}{u_M} =0 \quad \forall Q\in \mathcal{F} \right\}.
			\end{align*}
			The codimension of $\mathcal{H}_N$ is equal to $\vert \mathcal{F}\vert =N(C\mu)$. 
			
			We first claim that
			\begin{equation} \label{conjugation}
				\int_{\mathbb{R}^d} \left(\vert \nabla \varphi \vert^2 +  (V+M) \vert \varphi \vert^2 \right) 
				= \int_{\mathbb{R}^d} u_M^2 \bigg(\left \vert\nabla \left(\frac{\varphi}{u_M} \right) \right\vert^2 + \frac{1}{u_M}\left \vert \frac{\varphi}{u_M} \right  \vert^2\bigg)
			\end{equation}
			for all $\varphi\in H^1(\mathbb{R}^d) \cap L^\infty(\mathbb{R}^d)$ with compact support. As $u_M\in H^1_\mathrm{loc}(\mathbb{R}^d)$ and $1/u_M \in L^\infty_\mathrm{loc}(\mathbb{R}^d)$ by assumption, we get that $\varphi^2/u_M \in H^1(\mathbb{R}^d) \cap L^\infty(\mathbb{R}^d)$ with compact support and thus, after a simple approximation argument, we have
			\begin{align*}
				\int_{\mathbb{R}^d} \left( \nabla u_M \cdot \nabla \left(\frac{\vert \varphi\vert^2}{u_M}\right) + (V+M) u_M  \frac{\vert\varphi\vert^2}{u_M} \right) = \int_{\mathbb{R}^d} \frac{\vert \varphi \vert^2}{u_M},
			\end{align*}
			since $u_M$ solves~(\ref{eq:weak form}). Furthermore, using the product rule for Sobolev functions \cite[Lemma 7.4]{lieb2001analysis} yields
			\begin{align*}
				\nabla u_M \cdot \nabla \left(\frac{\vert \varphi\vert^2}{u_M}\right) = \vert \nabla \varphi \vert^2 - u_M^2 \left \vert\nabla \left(\frac{\varphi}{u_M} \right) \right\vert^2,
			\end{align*}
			which readily implies \eqref{conjugation}.
			As $C_c^\infty(\mathbb{R}^d)$ is dense in the form domain of $-\Delta +V+M$ (see \cite[Theorem A.2.8.]{simon1982schrodinger}) the statement of the lemma follows if we can show that there exists $\varepsilon>0$ such that for all non-zero $\varphi\in C_c^\infty(\mathbb{R}^d)  \cap \caH_N$, 
			\begin{align*}
\int_{\mathbb{R}^d}  \bigg(u_M^2 \left \vert\nabla \left(\frac{\varphi}{u_M} \right) \right\vert^2 +  \frac{\vert \varphi\vert^2}{u_M}  \bigg)
  > (1+\varepsilon)\mu \int_{\mathbb{R}^d} \vert \varphi \vert^2.
			\end{align*}
			The additional $\varepsilon$ is needed to ensure that the strict inequality (without $\varepsilon$) holds true on all of $\caH_N$.
			We check this inequality using the partition into boxes. In any box $Q\notin \mathcal{F}$, the bound $\inf_Q 1/u_M > \mu$ and the nonnegativity of the first term yield the claim. If $Q\in \mathcal{F}$, we use that $\varphi\in\caH_N$, namely that the integral of $\varphi/u_M$ vanishes, together with Poincar\'e inequality on $Q$ with optimal constant $\frac{\pi^2}{d}(C\mu)$, see \cite{payne1960optimal}, to conclude that
			\begin{align*}
				\int_Q u_M^2 \vert \nabla \left(\frac{\varphi}{u_M} \right) \vert^2
				\geq \left( \frac{\inf_Q u_M}{\sup_Q u_M} \right)^2 \frac{\pi^2 C \mu }{d} \int_{Q} \vert \varphi \vert^2.
			\end{align*}
			However, using the Harnack-Moser inequality \eqref{Harnack Moser inequality} and the definition of $\mathcal{F}$ we obtain
			\begin{align*}
				C_{HM} \inf_Q u_M \geq \sup_Q u_M - \frac{C_{HM}}{C\mu} \geq \left( 1 - \frac{C_{HM}}{C}  \right) \frac{1}{\mu}.
			\end{align*}
			Thus, $\inf_Q u_M \geq \frac{1}{2C_{HM}\mu}$ for any $C\geq 2C_{HM}$ and hence, again by~(\ref{Harnack Moser inequality})
			\begin{align*}
				\sup_Q u_M \leq C_{HM}\left( \inf_Q u_M + (C\mu)^{-1}\right)  \leq 2 C_{HM} \inf_Q u_M.
			\end{align*}
			This yields
			\begin{align*}
				\int_Q u_M^2 \vert \nabla \left(\frac{\varphi}{u_M} \right) \vert^2
				\geq \frac{C\pi^2}{4d C_{HM}^2} \mu \int_Q \vert \varphi \vert^2
			\end{align*}
			and so any choice of $C>\frac{4d C_{HM}^2}{\pi^2}$ will give the desired estimate.
		\end{proof}
	
		\begin{lemma}\label{lma:LowerBound}
			Let $V$ satisfy the conditions of Theorem~\ref{thm:main}. Then
			\begin{align*}
				n(\mu) \leq \mathcal{N}^{V+M}\left(C\mu\right)
			\end{align*}
			where $C=1+2^{d+2}C_{HM}^2$ and for all $\mu\in\bbR$. 
		\end{lemma}
		\begin{proof}

			For a lower bound $N\leq \mathcal{N}^V(C\mu)$ it suffices, again by the Min-Max Principle, to find a subspace $\mathcal{H}_N \subseteq \mathrm{dom}(H^{1/2})$ of dimension at least $N$ such that 
			\begin{align*}
				\int_{\mathbb{R}^d} \left( \vert \nabla v \vert^2 + (V+M) \vert v \vert^2 \right) \leq C\mu \int_{\mathbb{R}^d} \vert v \vert^2.
			\end{align*}
			We define
			\begin{align*}
				\mathcal{F} = \left\{ Q \in \mathcal{Q}_{\mu^{-1/2}} \ : \ \sup_Q \, \frac{1}{u_M} \leq \mu \right\}.
			\end{align*}
			Furthermore, for a box $Q$ we pick $\chi_Q \in C_c^\infty(Q)$ with $0\leq \chi_Q \leq 1, \Vert \nabla \chi_Q \Vert_{L^\infty(\mathbb{R}^d)} \leq \widetilde{C}\mu^{1/2}$ for some $\widetilde{C}>4$, $\chi_Q \equiv 1$ on $Q/2$. Since the functions $\chi_Q u_M$ are non-zero and orthogonal to each other, the space
			\begin{align*}
				\mathcal{H}_N = \mathrm{span}\{ \chi_Q u_M \ : \ Q \in \mathcal{F} \}
			\end{align*}
			is of dimension $\vert \mathcal{F} \vert = n(\mu)$. 
			
			By assumption we have $u_M\in H_\mathrm{loc}^1(\mathbb{R}^d)\cap L_\mathrm{loc}^\infty(\mathbb{R}^d)$ and thus $\chi_Q u_M$ is in $H^1(\mathbb{R}^d) \cap L^\infty(\mathbb{R}^d)$ with compact support. By~\eqref{conjugation} and an approximation argument for $\chi_Q u_M$, we get
			\begin{equation}\label{chi u Identity}
				\int_{\mathbb{R}^d} \left( \vert \nabla(\chi_Q u_M)\vert^2 + (V+M) \chi_Q^2 u_M^2 \right)
				= \int_{\mathbb{R}^d} \left( \chi_Q^2 u_M + \vert \nabla \chi_Q \vert^2 u_M^2 \right),
			\end{equation}
and hence, since $Q\in\caF$,			%
			\begin{align} \label{chi u UpperBound}
				\int_{\mathbb{R}^d} \left( \vert \nabla (\chi_Q u_M) \vert^2 + (V+M) \chi_Q^2 u_M^2 \right)
				& \leq \left(\sup_Q \, \frac{1}{u_M} \right) \int_{Q} \chi_Q^2 u_M^2 + \widetilde{C}^2 \mu  \int_{Q} u_M^2 \nonumber \\
				& \leq \mu \left( \int_{Q} \chi_Q^2 u_M^2 + \widetilde{C}^2\int_{Q} u_M^2 \right). 
			\end{align}
			We are left to replace the integrand of the second term by $(\chi_Q u_M)^2$. As $Q\in \mathcal{F}$ we have $\inf_Q u_M\geq 1/\mu$ and hence, by the Harnack-Moser inequality \eqref{Harnack Moser inequality} we get
			\begin{align*}
				\sup_Q u_M \leq C_{HM} \left( \inf_Q u_M + \frac{1}{\mu} \right) \leq 2C_{HM} \inf_Q u_M.
			\end{align*}
We conclude that
			\begin{equation*}
				\int_Q u_M^2 \leq \vert Q \vert \, \sup_Q \, u_M^2 \leq  \vert Q \vert \, (2C_{HM})^2 \, \inf_{Q/2} \, u_M^2 
				\leq 2^{d+2} C_{HM}^2 \int_{Q/2} u_M^2
				\leq 2^{d+2}C_{HM}^2 \int_{\mathbb{R}^d} \chi_Q^2u_M^2,
			\end{equation*}
			where the last inequality follows from the properties of $\chi_Q$. This yields 
			\begin{align*}
				n(\mu) \leq \caN^{V+M}((1+\widetilde{C}2^{d+2}C_{HM}^2) \mu).
			\end{align*}
			As $\mu \mapsto \caN^{V+M}(\mu)$ is right-continuous and any $\widetilde{C}>4$ is admissible, we obtain the claim.
		\end{proof}

\section{Kato-class potentials} \label{sect:existence}

We now turn to the case of potentials $V$ that are in the Kato-class $\mathcal{K}_d$. We say that a measurable function $V: \mathbb{R}^d \rightarrow \mathbb{R}$ is in the Kato-class $\mathcal{K}_d$ if
\begin{align} \label{KatoClass}
		\lim_{\varepsilon \rightarrow 0^+} \sup_{x\in \mathbb{R}^d} \int_{B(x,\varepsilon)} \vert x -y \vert^{-(d-2)} \vert V(y) \vert dy &=0, && d\geq 3,\notag \\ 
		\lim_{\varepsilon \rightarrow 0^+} \sup_{x\in \mathbb{R}^d} \int_{B(x,\varepsilon)} \ln(\vert x -y\vert^{-1}) \vert V(y) \vert dy &=0, && d=2, \\ 
		\sup_{x\in \mathbb{R}^d} \int_{B(x,1)} \vert V(y) \vert dy &<\infty, && d =1. \notag
\end{align}
We equip the sets $\caK_d$ with the following norms: 
\begin{align*}
	\Vert V \Vert_{\mathcal{K}_d} = \begin{cases}
		\displaystyle \sup_{x\in \mathbb{R}^d}  \int_{B(x,1)} \vert x -y \vert^{-(d-2)} \vert V(y) \vert dy, & \qquad d\geq 3, \\ 
		\displaystyle \sup_{x\in \mathbb{R}^d}  \int_{B(x,1/2)} \ln(\vert x -y \vert^{-1}) \vert V(y) \vert dy ,& \qquad d=2, \\
		\displaystyle \sup_{x\in \mathbb{R}^d}  \int_{B(x,1)} \vert V(y) \vert dy, & \qquad d =1.
	\end{cases}
\end{align*}
Furthermore, we say $V\in \caK_d^\mathrm{loc}$ if $\mathbbm{1}_{B(0,R)} V \in \caK_d$ for all $R>0$. 

We start by recalling some properties about the Green's function of $-\Delta +V$ in this particular case.

\begin{theorem}{\cite[Theorem B.7.2.]{simon1982schrodinger}} \label{thm:Greens function}
	Suppose $V_+\in \mathcal{K}_d^\mathrm{loc}, V_-\in \mathcal{K}_d$. Let $0<\alpha<d/2$ and $\mathrm{Re}(z)<\inf \mathrm{spec}(-\Delta+V)$ or $z \notin \mathrm{spec}(-\Delta +V)$, $\alpha$ integral. Then $(-\Delta+V-z)^{-\alpha}$ is an integral operator with integral kernel $G^{(\alpha)}(x,y;z)$ obeying
	\begin{enumerate}
		\item $G^{(\alpha)}(\cdot, \cdot, z)$ is continuous away from $x=y$ and bounded uniformly in each region $\{(x,y) \in \mathbb{R}^d \times \mathbb{R}^d \ : \ \vert x-y\vert \geq a \}$ for $a>0$. 
		\item We have
		\begin{equation} \label{upper bound Green fct polynom}
			\vert G^{(\alpha)}(x,y;z) \vert \leq C \vert x-y\vert^{-d+2\alpha}.
		\end{equation}
		\item For $\vert x-y\vert$ sufficiently small and $\vert x \vert < R$
		\begin{equation} \label{lower bound Green fct}
			\vert G^{(\alpha)}(x,y;z)\vert \geq \widetilde{C}_R \vert x -y \vert^{-d+2\alpha}.
		\end{equation}
		If $V\in \mathcal{K}_d$, $\widetilde{C}_R$ can be chosen independent of $R$.
		\item For $\vert x-y\vert\geq 1$ 
		\begin{equation} \label{upper bound Green fct}
			\vert G^{(\alpha)}(x,y;z) \vert \leq C_{\delta, \alpha,z} \exp(-\delta \vert x-y\vert)
		\end{equation}
		for some $\delta>0$ and if $\mathrm{Re}(z)<\inf \mathrm{spec}(-\Delta+V)$, for all $\delta$ with $\frac{1}{2} \delta^2 <\inf \mathrm{spec}(-\Delta+V)- \mathrm{Re}(z)$.
	\end{enumerate}
\end{theorem}

This allows us to prove next that the Green's function is nonnegative.

\begin{lemma} \label{lm:positivity}
	Under the assumptions of Proposition \ref{prop:existence} we have for $d\geq 3$ that $ G^{(1)}(x,y;-M) \geq 0$ and for all $f\in L^2(\mathbb{R}^d)$ we have that $f\geq 0$ implies $(-\Delta+V+M)^{-1}f\geq 0$.
\end{lemma}
\begin{proof}
	By Theorem \ref{thm:Greens function} we know that $G^{(1)}(\cdot, \cdot; -M)$ is continuous away from the diagonal. Hence, it is enough to prove that $(-\Delta+V+M)^{-1}$ maps nonnegative $L^2$ functions to nonnegative $L^2$ functions (otherwise consider a positive bump function localize where the Green's function is negative to get a contradiction).
	
	We note that for all $t>0$ we have that $(-\Delta + t)^{-1}$ maps nonnegative $L^2$ functions to nonnegative $L^2$ functions, which can easily be seen from the fact that the corresponding explicit integral kernel is positive. As $V_-$ is form bounded by the Laplacian with relative bound zero (see \cite[Proposition A.2.3]{simon1982schrodinger}) and $\mathcal{D}((-\Delta)^{1/2}) \cap \mathcal{D}(V_+^{1/2})$ is dense in $L^2(\mathbb{R}^d)$ (as $C_c^\infty(\mathbb{R}^d)$ is contained in this intersection), we get by \cite[Lemma 2.1, Lemma 2.2]{perelmuter1981positivity} applied to $V_+$ and $V_-$ that $(-\Delta+V+M)^{-1}$ maps nonnegative $L^2$ functions to nonnegative $L^2$ functions.
\end{proof}

\begin{proof}[Proof of Proposition \ref{prop:existence}]
	We start with the case $d\geq 3$. In the end we will use Hadamard's method of descent to obtain all the properties for $d\in \{1,2\}$ from the case $d=3$.
	
	By assumption $\inf \mathrm{spec}(-\Delta +V+M)>0$ and thus, the Green's function $G_M(x,y):=G^{(1)}(x,y;-M)$ satisfies the properties of Theorem \ref{thm:Greens function}. We define
	\begin{align*}
		u_M(x) = \int_{\mathbb{R}^d} G_M(x,y) dy
	\end{align*}
	and will check that this function satisfies indeed all the properties of Proposition \ref{prop:existence}.
	
	The boundedness of $u_M$ follows immediately from~\eqref{upper bound Green fct polynom} and~\eqref{upper bound Green fct}. In order to check that $u_M$ is a weak solution of \eqref{landscape1} we follow the proof of \cite[Proposition 3.3]{bachmann2023counting} and set $u_{M,L}$ to be the solutions of
	\begin{equation} \label{approximate solutions}
		(-\Delta+V+M) u_{M,L} = \mathbbm{1}_{B(0,L)}. 
	\end{equation}
	Since $\mathbbm{1}_{B(0,L)}$ is in $L^2(\mathbb{R}^d)$, Theorem~\ref{thm:Greens function} immediately implies that	
	\begin{align*}
		u_{M,L}(x) = \int_{\mathbb{R}^d} G_M(x,y) \mathbbm{1}_{B(0,L)}(y) dy.
	\end{align*}
	Hence, by dominated convergence we have $\lim_{L\rightarrow \infty} u_{M,L}(x)=u_M(x)$. As $u_{M,L}$ is a weak solution of \eqref{approximate solutions}, one readily checks that $u_M$ is a solution of \eqref{eq:weak form}. The fact that $1/u_M\in L^\infty_\mathrm{loc}(\mathbb{R}^d)$. This follows from the lower bound \eqref{lower bound Green fct} together with the fact that the Green's function $G_M$ is nonnegative by Lemma~\ref{lm:positivity}. As $1/u_M \in L_\mathrm{loc}^\infty(\mathbb{R}^d)$ we get $V+M+1/u_M \in \caK_d^\mathrm{loc}$. The function $u_M$ is nonnegative and is a distributional solution of $(-\Delta +V+M+1/u_M) u_M=0$. Thus, \cite[Theorem 1.5]{AizenmanSimon} yields that $u_M$ is continuous. Finally, $u_M\in H_\mathrm{loc}^1(\mathbb{R}^d)$ follows from a Caccioppoli-type inequality (see \cite[Lemma 1.1]{chiarenza1986harnack}).
	
The cases $d\in \{1,2\}$ will now be covered by Hadamard's method of descent. Let's start with $d=2$. We define $\widetilde{V}(y_1, y_2, y_3) = V(y_1, y_2)$ and check that $\widetilde{V}$ satisfies the conditions of Proposition \ref{prop:existence} for $d=3$. Once we have done that, we can use the same proof as in \cite[page 10]{bachmann2023counting} to conclude. Similarly, one deals with the case $d=1$.
	
	We first check that if $-\Delta_{\mathbb{R}^d} + V \geq M$, then we also have $-\Delta_{\mathbb{R}^3}+ \widetilde{V} \geq M$. To see this for $d=2$, we take a Fourier transform in the variables and note that $-\Delta_{\mathbb{R}^3} +\widetilde{V}$ is unitarily equivalent to the direct integral
	\begin{align*}
		\int_{\mathbb{R}}^\oplus (p^2+(-\Delta_{\mathbb{R}^2}+V)) dp.
	\end{align*}
	Thus, by \cite[Theorem XIII.85]{reed1978iv} we have $\text{spec}(-\Delta_{\mathbb{R}^3} + \widetilde{V}) = [\min \text{spec}(-\Delta_{\mathbb{R}^2} +V), \infty)$ and hence, establishes our first property.
	
	We are left to show that for $d\in \{1,2\}$ if $V\in \mathcal{K}_d$ (respectively $\mathcal{K}_d^\mathrm{loc}$), then $\widetilde{V} \in \mathcal{K}_3$ (respectively $\mathcal{K}_3^\mathrm{loc}$). Let $d=2$, $x\in \mathbb{R}^3$ and $0<\varepsilon\leq 1$. We compute
	
	\begin{align*}
		\int_{B(x,\varepsilon)} \frac{\vert \widetilde{V}(y)}{\vert x-y \vert} dy
		&\leq 2\int_{B((0,0), \varepsilon)} \vert V(z_1+x_1, z_2+x_2) \vert \left( \frac{\pi}{2}+ \ln \left(\frac{\varepsilon}{\sqrt{z_1^2+z_2^2}}\right)_+ \right) dz_2 dz_3 \\
		&\leq (\pi +2) \int_{B((x_1,x_2),\varepsilon)} \frac{\vert V(z_1, z_2)\vert}{\vert (x_1, x_2)- (z_1, z_1) \vert} dz_1 dz_2.
	\end{align*}
	
	For $d=1$ we define $\widetilde{V}(y_1, y_2) = V(y_1)$. We compute for $x\in \mathbb{R}^2$ and $0<\varepsilon\leq 1$
	\begin{align*}
		\int_{B(x,\varepsilon)} \vert \widetilde{V}(y)\vert \ln(\vert x -y \vert^{-1}) dy
		&\leq \int_{B(0,\varepsilon)} \vert V(x_1+y_1) \vert \left( \int_{-1}^1 \vert \ln(\vert y_2 \vert) \vert dy_2 \right) dy_1 \\
		&\leq 2\sup_{x_1\in \mathbb{R}} \int_{B(x_1, \varepsilon)} \vert V(y_1) \vert dy_1.
	\end{align*}
	
	Let $d=2$, let $\widetilde{V}(y_1, y_2, y_3) = V(y_1, y_2)$ and let $\widetilde{u}_M$ be the solution of 
	\begin{align*}
		(-\Delta_{\mathbb{R}^3}+ \widetilde{V}+M) \widetilde{u}_M = 1
	\end{align*}
	which we obtained in the case $d=3$. One readily checks that $\widetilde{v_{M,L}}(x,t)=\widetilde{u_{M,L}}(x,t+\alpha)$ is a weak solution of $(-\Delta +\widetilde{V}+M) \widetilde{v_{M,L}} = \mathbbm{1}_{B((0,0,\alpha),L)}$. Thus, we have
	\begin{align*}
		\widetilde{u_{M,L}}(x,t+\alpha) = \int_{\mathbb{R}^3} G_M((x,t), y) \mathbbm{1}_{B((0,0,\alpha),L)}(y) dy.
	\end{align*}
	By dominated convergence we get
	\begin{align*}
		\widetilde{u_M}(x,t+\alpha) = \lim_{L\rightarrow \infty} \widetilde{u_{M,L}}(x, t+\alpha) = \int_{\mathbb{R}^3} G_M((x,t), y) dy
		= \widetilde{u_M}(x,t).  
	\end{align*}
	We define $u_M(x) = \widetilde{u_M}(x,0)$. One readily checks that $u_M$ inherits all the desired properties from $\widetilde{u_M}$. The same strategy works for $d=1$.

Finally, we consider the case where $V_+\in \mathcal{K}_d$. In this case the lower bound \eqref{lower bound Green fct} holds true for all $x$ and hence $\inf_{\mathbb{R}^d} u_M>0$. This implies in turn that $0<A_M<\infty$ since we have already established that $\sup_{\mathbb{R}^d}u_M<\infty$.
\end{proof}

Next we turn to the proof of Theorem \ref{thm:Kato}.

\begin{proof}[Proof of Theorem \ref{thm:Kato}]
(i) We adapt the proof in Lemma~\ref{lma:LowerBound} to the potential $\frac{1}{u_M}-M$ and define
		\begin{align*}
			\mathcal{F} = \left\{ Q\in \mathcal{Q}_{\vert C_c\mu \vert^{-1/2}} \ : \ \sup_Q \left( \frac{1}{u_M}-M \right)\leq c\mu \right\}
		\end{align*}
Similarly, we pick a bump function $\chi_Q$ as before, but with $\Vert \nabla \chi_Q \Vert_{L^\infty(\mathbb{R}^d)}\leq 5\vert C_c \mu \vert^{1/2}$ and we let
		\begin{align*}
			\mathcal{H} = \mathrm{span}\left\{ \chi_Q u_M \ : \ Q\in \mathcal{F} \right\},
		\end{align*}
for which $\mathrm{dim}(\mathcal{H})=\vert \mathcal{F} \vert = n_c(\mu)$ holds. With this, (\ref{chi u Identity},\ref{chi u UpperBound}) become
		\begin{align*}
			\int_{\mathbb{R}^d} \left( \vert \nabla(\chi_Q u_M) \vert^2 + V(\chi_Q u_M)^2 \right)
			&= \int_{Q} \vert \nabla \chi_Q \vert^2 u_M^2 + \int_Q \left( \frac{1}{u_M}-M \right) (\chi_Q u_M)^2 \\
			&\leq  5^2 C_c \vert \mu \vert \int_Q u_M^2+ c\mu \int_{\mathbb{R}^d} (\chi_Q u_M)^2. 
		\end{align*}
Here, we use
\begin{equation*}
\int_Q u_M^2 \leq \vert Q\vert \sup_{Q} u_M^2 
\leq \vert Q\vert A_M^2 \inf_{\bbR^d} u_M^2
\leq \vert Q\vert A_M^2 \inf_{Q/2} u_M^2
\leq 2^d(A_M)^2 \int_Q (\chi_Q u_M)^2
\end{equation*}
to conclude that
		\begin{align*}
			\int_{\mathbb{R}^d} \left( \vert \nabla(\chi_Q u_M) \vert^2 + V(\chi_Q u_M)^2 \right) \leq \left(c - 2^d( 5 A_M)^2 C_c \right) \mu \int_{\mathbb{R}^d} (\chi_Q u_M)^2.	
		\end{align*}
		Recalling that $C_c= \frac{c-1}{2^d( 5 A_M)^2}$ we obtain that $\mathcal{N}^V(\mu)\geq \vert \mathcal{F}\vert=n_c(\mu)$. 
		
(ii) With the additional assumption~\eqref{scale invariant Harnack}, we have that
\begin{equation*}
n_c(\mu) \geq \left\vert \left\{ Q\in \mathcal{Q}_{(C_c\vert \mu \vert)^{-1/2}} \ : \ \inf_Q \left(\frac{1}{u_M}-M\right)\leq \frac{c\mu}{\widetilde{C}_H} \right\} \right\vert.
\end{equation*}
Just as in the proof of Lemma~\ref{lm:comparing}, this is now lower bounded by the measure of the sublevel set.

(iii) Let $\mu\in \mathbb{R}$ and $\varepsilon>0$. For every $v$ in the form domain of $-\Delta +V$ we have by~\eqref{conjugation}
		\begin{align*}
			\int_{\mathbb{R}^d} &\left( \vert \nabla v(x) \vert^2+(V(x)-\mu-\varepsilon) \vert v(x)\vert^2 \right) dx \\
			&= \int_{\mathbb{R}^d} \left( \left\vert \nabla \left(\frac{v(x)}{u_M(x)}\right)\right\vert^2 + \left( \frac{1}{u_M(x)}-M-\mu-\varepsilon \right)\left\vert \frac{v(x)}{u_M(x)}\right\vert^2 \right) u_M(x)^2 dx \\
			&\geq \left( \inf_{\mathbb{R}^d} u_M \right)^2 \int_{\mathbb{R}^d} \left(\left\vert \nabla \left( \frac{v(x)}{u_M(x)}\right)\right\vert^2 - A_M^2 \left(\frac{1}{u_M(x)}-M-\mu-\varepsilon\right)_- \left\vert \frac{v(x)}{u_M(x)}\right\vert^2  \right) dx.
		\end{align*}
		If the integral on the RHS of \eqref{CLR improved} is infinite, then there is nothing to prove. Otherwise, we know from the standard CLR bound~\cite{cwikel1977weak} that
		\begin{equation}\label{CLR 1/u}
			N:=\caN^{-A_M^2(1/u_M-M-\mu-\varepsilon)_-}\left( 0 \right) \leq C_\mathrm{CLR} A_M^d \int_{\mathbb{R}^d}  \left( \frac{1}{u_M(x)}-M-\mu-\varepsilon \right)_-^{d/2} dx
		\end{equation}
		and therefore, there exist linearly independent $\varphi_1, \dots, \varphi_N \in L^2(\mathbb{R}^d)$  such that for $0\neq w\in \{ f\in H^1(\mathbb{R}^d) \ \vert \ \forall j\in \{1, \dots, N \} \ : \ \langle f, \varphi_j \rangle_{L^2(\mathbb{R}^d)} =0 \}$ we have
		\begin{align*}
			\int_{\mathbb{R}^d} \left(\left\vert \nabla w(x)\right\vert^2 - A_M^2 \left(\frac{1}{u_M(x)}-M-\mu-\varepsilon\right)_- \left\vert w(x)\right\vert^2  \right) dx > 0.
		\end{align*} 
		Next, we define the codimension $N$ space
		\begin{align*}
			\caK=\left\{ v\in \mathrm{dom}((-\Delta+V)^{1/2}) \ \vert \ \forall j\in \{1, \dots, N\} \ : \ \langle v, \varphi_j u_M \rangle_{L^2(\mathbb{R}^d)} =0  \right\}.
		\end{align*}
	For all $0\neq v\in \caK \cap C_c^\infty(\mathbb{R}^d)$ we have by the above that
		\begin{align*}
			\int_{\mathbb{R}^d} &\left( \vert \nabla v(x) \vert^2+(V(x)-\mu-\varepsilon) \vert v(x)\vert^2 \right) dx> 0.
		\end{align*}
		As $\caK\cap C_c^\infty(\mathbb{R}^d)$ is dense $\caK$, we conclude that
		\begin{align*}
			\int_{\mathbb{R}^d} &\left( \vert \nabla v(x) \vert^2+V(x) \vert v(x)\vert^2 \right) dx  
			>\mu\int_{\mathbb{R}^d} \vert v(x)\vert^2 d
		\end{align*} 
		for all non-zero $v\in\caK$. Thus, by the min-max principle we get that $\caN^{V}( \mu)\leq N$. The claim follows from~(\ref{CLR 1/u}).
\end{proof}


\section{The ground state energy} \label{sect:groundstate}

In this Section we will establish the relationship between the bottom of the spectrum of $E_0 = \inf\mathrm{spec}(-\Delta +V)$ and the infimum of the  shifted landscape function. We follow~\cite[Lemma 5]{berg2022efficiency}, where a similar result is proved for the torsion function in the case of nonnegative potentials.

\begin{proof}[Proof of Proposition \ref{prop:groundstate}]
Let $-\Delta_L$ denote the Dirichlet Laplacian on $B(0,L)$ and let $H_L$ be the operator $-\Delta_L + V+M$ on $L^2(B(0,L))$. Note that $\sigma_L:=\inf \mathrm{spec}(H_L) \geq \inf \mathrm{spec}(-\Delta +V+M)>0$. Thus
	\begin{align*}
		v_L = H_L^{-1} \mathbbm{1}_{B(0,L)}
	\end{align*}
is well-defined. We remark that $v_L$ is nonnegative (integration by parts yields that the Dirichlet Laplacian satisfies the Beurling-Deny condition \cite[Theorem XIII.50]{reed1978iv} and hence $e^{t\Delta_L}$ is positivity-preserving for $t>0$ which implies via functional calculus that $H_L^{-1}$ is positivity-preserving \cite[Theorem XIII.44]{reed1978iv}).
	
	Kato-class potentials are infinitesimally form-bounded by the Dirichlet Laplacian (use \cite[Theorem 4.7]{AizenmanSimon} to get that $V$ is is infinitesimally form-bounded by $-\Delta$ on $\mathbb{R}^d$, which implies that the same holds true on the ball) and as the Dirichlet Laplacian has compact resolvent (\cite[Theorem XIII.73 and its corollary]{reed1978iv}), so does $H_L$ (by \cite[Theorem XIII.68]{reed1978iv}). In particular, $\sigma_L$ is an eigenvalue. Let $\varphi_L$ be an eigenfunction of $H_L$ with eigenvalue $\sigma_L$. We can choose $\varphi_L>0$ (see \cite[Theorem 11.8]{lieb2001analysis}). Furthermore, as $\varphi_L$ is square-integrable and $B(0,L)$ has finite Lebesgue measure, we get that $\varphi_L$ is integrable too. Thus, we get
	\begin{align*}
		\int_{B(0,L)} \varphi_L 
		= \int_{B(0,L)} \varphi_L (H_L v_L)
		= \sigma_L \int_{B(0,L)} \varphi_L v_L
		\leq \sigma_L \Vert v_L \Vert_{L^\infty(\bbR^d)} \int_{B(0,L)} \varphi_L.
	\end{align*}
	This implies
	\begin{align*}
		1 \leq \sigma_L \Vert v_L \Vert_{L^\infty(\bbR^d)}.
	\end{align*}
	Note that $(-\Delta+V+M)(u_M-v_L)=0$ on $B(0,L)$ and $u_M-v_L=u>0$ on $\partial B(0,L)$. Thus, by the weak maximum principle (see \cite[Theorem 8.1]{gilbarg1977elliptic}) we have $u_M>v_L$ on $B(0,L)$ and therefore,
	\begin{equation} \label{eq:lowerbndgroundstate}
		1\leq \sigma_L \Vert u_M \Vert_{L^\infty(\bbR^d)}.
	\end{equation}
	Note that $\mathrm{dom}(H_L^{1/2})\subseteq \mathrm{dom}(H_K^{1/2}) \subseteq \mathrm{dom}((-\Delta+V+M)^{1/2})$ for all $K\geq L\geq 0$. Thus, $L\mapsto\sigma_L$ is a decreasing function and so
	\begin{align*}
		\lim_{L\rightarrow \infty} \sigma_L \geq \inf \mathrm{spec}(-\Delta+V+M).
	\end{align*}
	 Since $C_c^\infty(\mathbb{R}^d)$ is dense in the form domain of $-\Delta +V+M$ (see \cite[Theorem 8.2.1]{davies1995spectral}, there is a sequence $(\xi_n)_{n\in \mathbb{N}} \subseteq C_c^\infty(\mathbb{R}^d)$, normalized in $L^2(\mathbb{R}^d)$, such that
	 \begin{align*}
	 	\lim_{n\rightarrow \infty} \langle \xi_n, H \xi_n \rangle_{L^2(\mathbb{R}^d)} = \inf \mathrm{spec}(-\Delta+V+M).
	 \end{align*}
 	Now, for every $n$ there exists $L_n$ such that $\xi_n \in \mathrm{dom}(H_{L_n}^{1/2})$ and so $\langle \xi_n, H \xi_n \rangle_{L^2(\mathbb{R}^d)} \geq \sigma_{L_n}$. It follows that
 	\begin{align*}
 		\lim_{L\rightarrow \infty} \sigma_L = \inf \mathrm{spec}(-\Delta+V+M).
 	\end{align*}
 	Taking the limit $L\rightarrow \infty$ in \eqref{eq:lowerbndgroundstate} yields \eqref{eq:ground state}.
\end{proof}

We note that for nonnegative potentials, one can use the same proof as in \cite[Theorem 1]{van2009hardy}, respectively \cite[Lemma 5]{berg2022efficiency} to get the upper bound
\begin{equation*}
E_0\leq (4+d\log(2)) \inf_{\mathbb{R}^d} \frac{1}{u_M} - M.
\end{equation*}
Under additional assumption on the heat kernel of $-\Delta+V$ the dimensional constant can be further improved \cite[Theorem 1.5]{vogt2019estimates}.

The difference between the bottom of the spectrum of $-\Delta+V$ and the infimum of $1/u_M-M$ may a priori be large. However, the bound just proven allows us to set up the iteration procedure described in Corollary \ref{cor:iteration}  which we now show to converges to the bottom of the spectrum.

\begin{proof}[Proof of Corollary \ref{cor:iteration}]
	By \eqref{eq:ground state} we have that the function 
	\begin{align*}
		F: \quad(-E_0, \infty) &\rightarrow (-E_0, \infty) \\ M &\mapsto M - \inf_{\mathbb{R}^d} \frac{1}{u_M}
	\end{align*}
	is well defined. Furthermore, $u_M$ is positive and bounded above and so 
	\begin{align*}
		F(M) = M - \inf_{\mathbb{R}^d} \frac{1}{u_M} < M.
	\end{align*}
	This readily implies that the sequence $\{M_n\}_n$ is convergent with $M_\infty = \lim_{n\rightarrow \infty} M_n \geq -E_0$. We claim that $F$ is continuous, which implies that $F(M_\infty) = M_\infty$. If $M_\infty>-E_0$, then it is in the domain of $F$ and so $F(M_\infty)<M_\infty$, which is a contradiction. The continuity of $F$ follows from $u_{M+\delta} \leq u_M$ for all $\delta \geq 0$ (see \cite[Lemma 2.1]{perelmuter1981positivity}) and the fact that $\lim_{\delta \rightarrow 0^+} u_{M+\delta}(x) = u_M(x)$ for all $x\in \mathbb{R}^d$ (using \cite[Theorem 1.3]{perelmuter1981positivity} and the continuity of the shifted landscape functions).
	
Finally,
	\begin{align*}
		-E_0 < M_n - \inf_{\mathbb{R}^d} \frac{1}{u_{M_n}} < M_n
	\end{align*}
for	all $n$. Since $M_n\to -E_0$, we immediately conclude that $\lim_{n\to\infty}\inf_{\mathbb{R}^d} \frac{1}{u_{M_n}} = 0$, which implies~\eqref{eq:blowup}.
\end{proof}


\section{Examples} \label{sect:examples}

Since the constants that appear in the previous sections all depend on either Harnack constants $C_{HM}$ or $A_M$, it is a priori unclear how sharp the estimates in terms of the landscape function really are. In this section, we provide a partial answer to that question by considering special cases. First of all, we show that for radially symmetric potentials of the form $-\vert v\vert^{-\rho}$ (in particular the Coulomb potential) whose eigenvalues accumulate at the bottom of the essential spectrum $\mu = 0$, the sharp constant in the bound~(\ref{CLR improved}) can be computed so that it becomes asymptotically exact. Secondly, we turn to the other end of the discrete spectrum and analyze the first few eigenvalues appearing in case of the spherical well potential $-\varepsilon \mathbbm{1}_{\mathbb{R}^d \setminus B(0,\delta)}$ for $\varepsilon\ll1$: The landscape function is explicit and its minimum approximates the ground state energy to order $\varepsilon$.

\subsection{Asymptotics at the bottom of the essential spectrum}

The asymptotic for power law potentials is known and we first recall a general result~\cite{brownell1961asymptotic} in $d\geq 3$; the result extends to lower dimensions with appropriate assumptions. For other variants see also \cite[Theorem XIII.82]{reed1978iv} or (with an oscillatory prefactor) \cite[Theorem 3.1, Theorem 3.2]{raikov2016discrete}.

\begin{theorem}  \label{thm:BrownellClark}
	Let $d\geq3$.   Let $1<R<\infty$ and $V: \mathbb{R}^d \rightarrow \mathbb{R}$ such that $V$ is $C^1$ for $\vert x \vert \geq R$ with $\lim_{\vert x \vert \rightarrow \infty} V(x)=0$. There exists $q: \mathbb{R} \rightarrow (-\infty, 0]$ such that 
		\begin{equation} \label{radial conditon}
			V(x) = q(\vert x \vert) + O(\vert x \vert^{-2})\quad\text{as}\quad \vert x \vert \rightarrow \infty.
		\end{equation} 
		Moreover, $q$ is $C^5$ for $r \geq R$ and
		\begin{equation} \label{derivative condition basic}
			M_1 r^{-1} \leq \vert q^{(k)}(r)/q^{(k-1)}(r) \vert \leq M_2 r^{-1} \quad (k\leq 5)
		\end{equation}
for all $r\geq R$ and for some $0<M_1, M_2 <\infty$. Finally, we assume that there exists $\alpha >0$ such that $d+\alpha \geq 2$ and such that 
	\begin{equation} \label{integrability condition}
		\int_{\vert x \vert \leq R} \vert V(x) \vert^{(d+\alpha)/2} dx < \infty.
	\end{equation}
	Then, as $\mu \rightarrow 0^-$
	\begin{equation}
		\caN^{V}(\mu) = (1+o(1)) \left[ (2\sqrt{\pi})^d \Gamma\left( \frac{d}{2}+1 \right) \right]^{-1} \int_{\mathbb{R}^d} \left( \mu - V(x) \right)_+^{d/2} dx.
	\end{equation}
\end{theorem}

		\begin{proof}[Proof of Proposition \ref{prop:powerlaw}]
Our goal is to establish pointwise estimates for $u_M$ that match the assumptions~(\ref{radial conditon},\ref{derivative condition basic}). In particular, we will construct  a radial function $q$ such that $\vert u_M(x) - q(\vert x \vert)\vert \leq C_M/(1+\vert x \vert^2)$.

	By standard elliptic theory, we know that $u_M \in C^\infty(\mathbb{R}^d\setminus \{0\})$. The strategy is to compute successive approximations of $u_M$ in terms of powers of the potential. We first note that
	\begin{align*}
		(-\Delta + V +M) (u_M-\frac{1}{M}+\frac{V}{M^2}) = -\frac{\Delta V}{M^2} + \frac{V^2}{M^2},
	\end{align*}
	where $\Delta V$ decays fast, while $V^2$ does not. This formula suggests the correct expansion away from the origin where the potential is singular. Let $N\in \mathbb{N}$ be such that $N\rho <2 \leq (N+1)\rho$. Then	
	\begin{equation} \label{formal}
		(-\Delta + V +M) \bigg(u_M-\sum_{j=0}^N \frac{(-V)^j}{M^{j+1}}\bigg) = -\sum_{j=0}^N \Delta \bigg( \frac{(-V)^j}{M^{j+1}} \bigg) + (-1)^{N+1} \frac{V^{N+1}}{M^{N+1}},
	\end{equation}
	since $(V+M)\sum_{j=0}^N \frac{(-V)^j}{M^{j+1}})$ is telescopic. By assumption, the r.h.s.\ decays as $O(1/\vert x \vert^2)$. Although the resolvent is a bounded operator from $L^\infty(\mathbb{R}^d)$ to itself, this is not quite sufficient to conclude, both because of the singularity at the origin and the insufficient decay at infinity. We shall therefore introduce cutoffs in neighbourhoods of the origin and of infinity.
		
	We choose a radial function $\chi\in C_c^\infty(\mathbb{R}^d)$ with $0\leq \chi \leq 1$, $\chi(x)\equiv 1$ for $1/2\leq \vert x \vert \leq 3/2$ and $\chi(x) \equiv 0$ for $\vert x \vert \in [0,1/4] \cup [2,\infty)$. Let $x_0\in \mathbb{R}^d\setminus \{0\}$ and let
	$$ \tilde{\chi}(x) = \chi(\vert x_0\vert^{-1} x).$$
	With this, we define
	\begin{align*}
		\widetilde{u_M} = (-\Delta +V\tilde{\chi}+M)^{-1} \tilde{\chi}.
	\end{align*}
	Let moreover
	\begin{align*}
		u_{M,L} = (-\Delta +V+M)^{-1} \mathbbm{1}_{B(0,L)}.
	\end{align*}
	Now, 
	\begin{equation} \label{tilde u}
		(-\Delta + V\tilde{\chi} +M) \bigg(\widetilde{u_M}-\sum_{j=0}^N \frac{(-V)^j \tilde{\chi}^{j}}{M^{j+1}}\bigg) = -\sum_{j=1}^N \Delta \bigg( \frac{(-V)^j\tilde{\chi}^{j}}{M^{j+1}} \bigg) - \frac{(-V)^{N+1}\tilde{\chi}^{N+1}}{M^{N+1}},
	\end{equation} 
	which readily implies that 
	\begin{equation} \label{utilde-potential}
		\bigg\Vert \widetilde{u_M}-\sum_{j=0}^N \frac{(-V)^j \tilde{\chi}^{j}}{M^{j+1}} \bigg\Vert_{L^\infty(\mathbb{R}^d)} \leq \frac{C}{\vert x_0\vert^{2}},
	\end{equation}
	where $C$ only depend on $\rho, d, M, \Vert \chi \Vert_{C^2(\mathbb{R}^d)}$ and $\Vert (-\Delta +V+M)^{-1} \Vert_{\infty, \infty}$, but not on $x_0$.
	\newline	
We now need to estimate the difference $u_{M,L}(x_0)-\widetilde{u_M}(x_0)$. We have
	\begin{equation}\label{u-tilde u}
		(-\Delta + V\tilde{\chi}+M)(u_{M,L}-\widetilde{u_M}) = V(\tilde{\chi}-1)u_{M,L}+\mathbbm{1}_{B(0,L)} -\tilde{\chi} 
	\end{equation}
which is nonnegative whenever $L\geq 2\vert x_0 \vert$. Since $(-\Delta +V\tilde{\chi}+M)^{-1}$ is positivity-preserving, this implies that $u_{M,L} -\widetilde{u_M}\geq 0$ for $L\geq 2\vert x_0 \vert$.  Furthermore, if additionally $\vert x_0\vert \geq 2$, then we have
	\begin{align*}
		(-\Delta + V\tilde{\chi} +M)(u_{M,L}-\widetilde{u_M}) = 0
	\end{align*}
	on the annulus $B(0, \vert x_0\vert+1)\setminus B(0,\vert x_0\vert-1)$. Note that $u_{M,L}-\widetilde{u_M}$ is radial and, by the weak maximum principle, monotone in the radius and so
\begin{equation}\label{L2 to pointwise}
\vert u_{M,L}(x_0)-\widetilde{u_M}(x_0) \vert \leq \vert B(0,1/2)\vert^{-1/2} \Vert u_{M,L} - \widetilde{u_M} \Vert_{L^2(B(x_0,1))}.
\end{equation}
	Hence, we are left to obtain a local $L^2$ estimate in a neighbourhood of $x_0$. For this we will use yet another cut-off function. Namely, we choose $\Xi\in C_c^\infty(\mathbb{R})$ such that $\frac{1}{e^2} \leq \Xi(r) \leq 1$ for $\vert r \vert \leq 2$, $\Xi(r)=1$ for $\vert r \vert \leq 1$ and $\Xi(r)=e^{-r}$ for $r>2$. Then we define 
	\begin{align*}
		\Xi_\varepsilon(x) = \Xi(\varepsilon (\vert x -x_0\vert)),
	\end{align*}
	where we will choose $0<\varepsilon\leq 1$ later to be sufficiently small. We compute
	\begin{align*}
		\nabla \Xi_\varepsilon(x) = \Xi'(\varepsilon \vert x -x_0\vert) \varepsilon \frac{x-x_0}{\vert x -x_0 \vert}
	\end{align*}
	and
	\begin{align*}
		\Delta \Xi_\varepsilon(x) = \Xi''(\varepsilon \vert x -x_0\vert) \varepsilon^2 + \varepsilon\frac{d-1}{\vert x -x_0\vert} \Xi'(\varepsilon \vert x-x_0\vert). 
	\end{align*}
	Thus, for $\vert x-x_0\vert \geq 2/\varepsilon$ we have
	\begin{align*}
		\vert \nabla \Xi_\varepsilon(x)\vert +\vert \Delta \Xi_\varepsilon(x) \vert \leq \varepsilon\left[1+ \varepsilon\left( 1+ \frac{d-1}{2} \right)\right] \Xi_\varepsilon(x).
	\end{align*}
	As $1/e^2\leq \Xi(r) \leq 1$ for $\vert r \vert \leq 2$, we get for all $x\in \mathbb{R}^d$
	\begin{equation} \label{Xi}
		\vert \nabla \Xi_\varepsilon(x) \vert + \vert \Delta \Xi_\varepsilon(x) \vert
		\leq \varepsilon\left[ e^2\left(\Vert \Xi' \Vert_{L^\infty([0,2])}+ \Vert \Xi'' \Vert_{L^\infty([0,2])} \right) + \left( 2+ \frac{d-1}{2}\right) \right] \Xi_\varepsilon(x)= C \varepsilon \Xi_\varepsilon(x),
	\end{equation}
	with $C$ independent of $\varepsilon, x_0,x$.
	
With these estimates in hand, we go back to~(\ref{u-tilde u}). 	We will write $\Vert \cdot \Vert_{p,q}$ for the operator norm $\Vert \cdot \Vert_{\caL(L^p(\mathbb{R}^d), L^q(\mathbb{R}^d))}$. Furthermore, we will repeatedly use the fact (see \cite[Lemma 2.1]{perelmuter1981positivity}) that for all $V\leq W\leq 0$ and all $f\in L^2(\mathbb{R}^d)$, 
	\begin{equation*} 
		(-\Delta+W+M)^{-1} \vert f \vert \leq (-\Delta + V+M)^{-1} \vert f \vert,
	\end{equation*}
	where the inequality has to be understood pointwise almost everywhere. This implies in particular the operator monotonicity of $V\mapsto (-\Delta + V +M )^{-1}$ in $L^2(\bbR^d)$. In turn, the same holds by density for $\Vert(-\Delta + V +M )^{-1}\Vert_{1,2}$.

Commuting the multiplication operator $\Xi_\varepsilon$ through the resolvent yields	
	\begin{multline}\label{Estimate with cutoff}
		\Xi_\varepsilon (u_{M,L}-\widetilde{u_M}) 
		= (-\Delta +V\tilde{\chi}+M)^{-1} \Xi_\varepsilon (V(\tilde{\chi}-1)u_{M,L}+\mathbbm{1}_{B(0,L)} -\tilde{\chi}) \\
		\quad+ (-\Delta +V\tilde{\chi}+M)^{-1} \left[-2\nabla\cdot(\nabla  \Xi_\varepsilon) + (\Delta \Xi_\varepsilon) \right] (u_{M,L}-\widetilde{u_M}).
	\end{multline}
	Using \eqref{Xi} we can estimate the last term in $L^2$ by
	\begin{multline}\label{bound on commutator}
		\Vert (-\Delta +V\tilde{\chi}+M)^{-1} \left[-2  \nabla\cdot(\nabla \Xi_\varepsilon) + (\Delta \Xi_\varepsilon) \right] (u_{M,L}-\widetilde{u_M}) \Vert_{L^2(\mathbb{R}^d)} \\
		\leq \left[2\Vert (-\Delta +V\tilde{\chi}+M)^{-1} \nabla \Vert_{2,2} + \Vert (-\Delta +V\tilde{\chi}+M)^{-1} \Vert_{2,2} \right] C\varepsilon \Vert \Xi_\varepsilon(u_{M,L}-\widetilde{u_M}) \Vert_{L^2(\mathbb{R}^d)}.
	\end{multline}

Since $V\leq V\tilde\chi\leq 0$, and with the fact that $V$ is form bounded by $-\Delta$ with relative bound zero, we conclude that
	\begin{align*}
		\Vert (-\Delta + V\tilde{\chi} +M)^{-1} \nabla \Vert_{2,2}
		\leq \Vert (-\Delta + V +M)^{-1/2} \Vert_{2,2} \Vert (-\Delta+V+M)^{-1/2} \nabla \Vert_{2,2} < \infty.
	\end{align*}
	Hence, the term of~(\ref{bound on commutator}) in $[\cdots]$ is bounded and there is $0<\varepsilon\leq 1$ such that
\begin{equation*}
\Vert (-\Delta +V\tilde{\chi}+M)^{-1} \left[-2  \nabla\cdot(\nabla \Xi_\varepsilon) + (\Delta \Xi_\varepsilon) \right] (u_{M,L}-\widetilde{u_M}) \Vert_{L^2(\mathbb{R}^d)}
\leq \frac{1}{2}\Vert \Xi_\varepsilon(u_{M,L}-\widetilde{u_M}) \Vert_{L^2(\mathbb{R}^d)}.
\end{equation*}
With this, (\ref{Estimate with cutoff}) yields
\begin{equation*}
\Vert \Xi_\varepsilon (u_{M,L}-\widetilde{u_M}) \Vert_{L^2(\mathbb{R}^d)} 
		\leq 2\Vert (-\Delta +V\tilde{\chi}+M)^{-1} \Xi_\varepsilon (V(\tilde{\chi}-1)u_{M,L}+\mathbbm{1}_{B(0,L)} -\tilde{\chi}) \Vert_{L^2(\mathbb{R}^d)}
\end{equation*}
We now split $V(\tilde{\chi}-1)u_{M,L} = V(\tilde{\chi}-1) (\mathbbm{1}_{\mathbb{R}^d\setminus B(0,1)} + \mathbbm{1}_{B(0,1)})u_{M,L}$. Using that 
	$$\vert V(\tilde{\chi}-1) \mathbbm{1}_{\mathbb{R}^d\setminus B(0,1)}u_{M,L}+ \mathbbm{1}_{B(0,L)}-\tilde{\chi} \vert \leq \mathbbm{1}_{\mathbb{R}^d \setminus B(x_0, \vert x_0\vert/2)},$$ 
	the definition of $\tilde{\chi}$ and the fact that $ (-\Delta+V+M)^{-1}$ is bounded from $L^1(\mathbb{R}^d)$ to $L^2(\mathbb{R}^d)$ (see \cite[Theorem B.2.2.]{simon1982schrodinger}) we conclude that
	\begin{align*}
		\Vert \Xi_\varepsilon (u_{M,L} -\widetilde{u_M}) \Vert_{L^2(\mathbb{R}^d)}
		&\leq 2\Vert (-\Delta +V+M)^{-1} \Vert_{2,2} \Vert  \Xi_\varepsilon \mathbbm{1}_{\mathbb{R}^d \setminus B(x_0, \vert x_0\vert/2)} \Vert_{L^2(\mathbb{R}^d)} \\
		&\quad +2 \Vert (-\Delta+V+M)^{-1} \Vert_{1,2}  \Vert \Xi_\varepsilon u_{M,L} \Vert_{L^\infty(B(0,1))} \Vert V \Vert_{L^1(B(0,1))} \\ 
		&\leq C e^{-\varepsilon \vert x_0 \vert/4},
	\end{align*}
because of our choice of $\Xi_\varepsilon$, where $C$ does not depend on $x_0$ or $L$.
Using~(\ref{L2 to pointwise}), the fact that $1/e^2\leq \Xi_\varepsilon(x) \leq 1$ for $x\in B(x_0,1)$, we finally obtain for $\vert x_0\vert \geq 4/\varepsilon$
\begin{align*}
	\vert u_{M,L}(x_0)-\widetilde{u}(x_0) \vert \leq Ce^{-\varepsilon \vert x_0 \vert/4},
\end{align*}
where $\varepsilon, C$ do not depend on $L$ or $x_0$ (but on $M$). Combining this with \eqref{utilde-potential} and noting that $\tilde{\chi}=1$ in a neighbourhood of $x_0$ we get for $\vert x_0\vert$ sufficiently large (independent of $L$)
\begin{align*}
	\vert u_{M,L}(x_0)-\sum_{j=0}^N \frac{(-V(x_0))^j}{M^{j+1}} \vert \leq \frac{C}{\vert x_0 \vert^2},
\end{align*}
with $C$ independent of $L, x_0$.

As $u_{M,L}$ converges pointwise to $u_M$ for $L\rightarrow \infty$, we conclude that
\begin{align*}
	\vert u_M(x)-\sum_{j=0}^N \frac{(-V(x))^j}{M^{j+1}} \vert \leq \frac{C}{\vert x \vert^2},
\end{align*}
for all $\vert x \vert$ sufficiently large. Recalling now that $-V(x) = \vert x\vert^{-\rho}$, this readily implies that there exists a polynomial $P$ with real coefficients of degree at most $N$ and $P(0)=0=P'(0)$ such that 
\begin{equation}\label{final estimate on atomic uM}
	\frac{1}{u_M(x)}-M-V(x)-P(\vert x \vert^{-\rho}) = O(\vert x \vert^{-2})
\end{equation}
for $\vert x \vert \rightarrow \infty$. Moreover, the boundedness of $1/u_M$ tells us that $\frac{1}{u_M(x)}-M$ satisfies \eqref{integrability condition}.

Altogether, we obtain first that
\begin{equation*}
	\caN^{V}(\mu) = (1+o(1)) \left[ (2\sqrt{\pi})^d \Gamma\left( \frac{d}{2}+1 \right) \right]^{-1} \int_{\mathbb{R}^d} \left( \mu - V(x) \right)_+^{d/2} dx
\end{equation*}
as $\mu\to 0^-$, since $V$ clearly satisfies the assumptions of Theorem~\ref{thm:BrownellClark}. By~(\ref{final estimate on atomic uM}), this further implies that
\begin{equation*}
	\caN^{V}(\mu) = (1+o(1)) \left[ (2\sqrt{\pi})^d \Gamma\left( \frac{d}{2}+1 \right) \right]^{-1} \int_{\mathbb{R}^d} \left( \mu - \left(\frac{1}{u_M(x)}-M\right) \right)_+^{d/2} dx.
\end{equation*}
Finally, (\ref{final estimate on atomic uM}) also implies that Theorem~\ref{thm:BrownellClark} applies to $\frac{1}{u_M}-M$ and therefore that
\begin{equation*}
	\caN^{V}(\mu) = (1+o(1)) \caN^{1/u_M-M}(\mu),
\end{equation*}
which concludes the proof.
\end{proof}		

In the case of the Coulomb potential, the proof above yields that the effective potential asymptotically equals the original potential itself, namely
\begin{equation*}
\frac{1}{u_M(x)} - M = -\frac{1}{\vert x\vert} + O(\vert x\vert^{-2}).
\end{equation*}

\subsection{Ground state energy for the square well}

We now turn to $V(x)=-\varepsilon \mathbbm{1}_{B(0,\delta)}$ for $\varepsilon, \delta>0$ and focus on the case $d=1$. The ground state energy of $-\Delta+\varepsilon \mathbbm{1}_{\mathbb{R}\setminus (-\delta,\delta)}$ is, for $\varepsilon$ small enough, given by the smallest positive solution of the transcendental equation
\begin{align*}
	\sqrt{\varepsilon-x}=\sqrt{x} \tan(\sqrt{x} \delta).
\end{align*}
This can be rewritten as
\begin{align*}
	\varepsilon = F(x),\quad\text{where}\quad F(x) = x \left( 1+ \tan(\sqrt{x} \delta)^2\right).
\end{align*}
The function $F$ has an analytic extension at the origin such that $F(x) = x+ x^2 \delta^2 + O(x^3)$, and hence
\begin{align*}
	x=F^{-1}(\varepsilon) = \varepsilon - \delta^2 \varepsilon^2 + O(\varepsilon^3).
\end{align*}
This implies that as $\varepsilon \to 0^+$
\begin{align*}
	\inf \mathrm{spec}(-\Delta - \varepsilon \mathbbm{1}_{(-\delta, \delta)})
	=  \inf \mathrm{spec}(-\Delta + \varepsilon \mathbbm{1}_{\mathbb{R}\setminus(-\delta, \delta)}) - \varepsilon
	= -\delta^2 \varepsilon^2 + O(\varepsilon^3).
\end{align*}

Next we compute the shifted landscape function, which is an even solution of
\begin{equation*}
-f''(x) + cf(x) = 1
\end{equation*}
where the constant $c$ depends on $M$ and whether $x$ is inside or outside of the well. 

We start with the case $M>\varepsilon$, in which case
\begin{align*}
	u_M(x) = \begin{cases}
		\frac{1}{M-\varepsilon}  + a_2 \cosh(\sqrt{M-\varepsilon}x),& \vert x \vert < \delta,\\
		\frac{1}{M} + b_1 e^{-\sqrt{M} \vert x \vert} + b_2 \cosh(\sqrt{M}x),& \vert x \vert> \delta.
	\end{cases}
\end{align*}
As $u_M$ is bounded, we get $b_2=0$. The other constants are set by the of $u_M$ differentiability at $\vert x \vert = \delta$:
\begin{align*}
	a_2 &= \frac{\frac{1}{M}-\frac{1}{M-\varepsilon}}{\cosh(\sqrt{M-\varepsilon}\delta)+\frac{\sqrt{M-\varepsilon}}{\sqrt{M}}\sinh(\sqrt{M-\varepsilon}\delta)} \\
	b_1 &= -\frac{\sqrt{M-\varepsilon}}{\sqrt{M}} \sinh(\sqrt{M-\varepsilon}\delta) e^{\sqrt{M}\delta} a_2.
\end{align*}
The maximum of $u_M$ is attained at the origin and so
\begin{equation*}
\inf_{\mathbb{R}} \frac{1}{u_M}-M 
=\varepsilon \frac{M-\varepsilon}{M} \frac{1}{\cosh(\sqrt{M-\varepsilon}\delta) + \frac{\sqrt{M-\varepsilon}}{\sqrt{M}} \sinh(\sqrt{M-\varepsilon}\delta)} -\varepsilon + o(\varepsilon).
\end{equation*}
Hence, for any $M>\epsilon$ --- in particular arbitrarily large! ---  the first approximation of the ground state energy given by~Corollary \ref{cor:iteration} is already of order $\varepsilon$.

For $M=\varepsilon$, the computation in~\cite[Section 6]{bachmann2023counting} yields
\begin{align*}
	\inf_{\mathbb{R}} \frac{1}{u_\varepsilon}-\varepsilon
	= \frac{-\delta \varepsilon^{3/2}-\frac{\delta^2 \varepsilon^2}{2}}{1+\delta \sqrt{\varepsilon} +\frac{\delta^2 \varepsilon}{2}}.
\end{align*}
Hence, for small $\varepsilon$, this yields an approximation of order $\varepsilon^{3/2}$.

Finally if $M<\varepsilon$, then
\begin{align*}
	u_M(x) = \begin{cases}
		\frac{1}{M-\varepsilon} + a \cos(\sqrt{\varepsilon-M} x),& \vert x \vert < \delta,\\
		\frac{1}{M} + b e^{-\sqrt{M} \vert x \vert},& \vert x \vert> \delta.
	\end{cases}
\end{align*}
Differentiability yields now
\begin{align}
	a &= \frac{\frac{1}{M}- \frac{1}{M-\varepsilon}}{\cos(\sqrt{\varepsilon-M}\delta)-\frac{\sqrt{\varepsilon-M}}{\sqrt{M}}\sin(\sqrt{\varepsilon-M}\delta)} \label{a} \\
	b&= \frac{\sqrt{\varepsilon-M}}{\sqrt{M}} \sin(\sqrt{\varepsilon-M}\delta) e^{\sqrt{M}\delta}a \label{b} 
\end{align}
This allows us to exhibit explicitly the blowup of~(\ref{eq:blowup}) in Corollary~\ref{cor:iteration}. For this we first note that $\varepsilon+E_0$ is the ground state energy of the operator $-\Delta+\varepsilon \mathbbm{1}_{(-\infty,-\delta)\cup (\delta,\infty)}(x)$ and hence
\begin{align*}
	\sqrt{-E_0} = \sqrt{\varepsilon+E_0} \tan(\sqrt{\varepsilon+E_0}\delta).
\end{align*}
It follows from this and~(\ref{a}) that $\lim_{M\to (-E_0)^+} a = \infty$ and the same for $b$ by~(\ref{b}). Hence $\lim_{M\to (-E_0)^+} u_M(x) = \infty$, or in other words for the effective potential
\begin{equation*}
\lim_{M\to (-E_0)^+}\left(\frac{1}{u_M(x)}-M\right) = E_0
\end{equation*}
for all $x\in\bbR$.
		
		\section{Upper and lower Lieb-Thirring bounds}\label{sect:LT}
		
		In this section we recall the standard argument for the derivation of the Lieb-Thirring inequality from estimates on the eigenvalue counting function. First we rewrite the sum of eigenvalues in terms of an integral.
		
		\begin{lemma} \label{lm: sum to integral}
			For any $\gamma>0$ we have
			\begin{equation} \label{eq:sumtointegral}
				\mathrm{tr}((-\Delta +V)_-^\gamma)= \gamma \int_0^\infty \lambda^{\gamma -1} \mathcal{N}^V(-\lambda) d\lambda.
			\end{equation}
		\end{lemma}
		\begin{proof}
			If $\mathcal{N}^V(-\lambda)=\infty$ for any $\lambda>0$, then both expression in \eqref{eq:sumtointegral} are infinity. On the other hand, if $\mathcal{N}^V(-\lambda)<\infty$ for all $\lambda>0$, then the negative spectrum consists of countably many eigenvalues of finite multiplicity. We label them as $\lambda_1\leq \lambda_2 \leq ... <0$. Then we can write
			\begin{align*}
				\mathrm{tr}((-\Delta+V)_-^\gamma) = \sum_{j} \vert \lambda_j \vert^\gamma.
			\end{align*}
Since $\mathcal{N}^V$ is a piecewise constant function, we have for any $N\in \mathbb{N}$
\begin{equation*}
\gamma \int_{\vert \lambda_N \vert}^\infty \lambda^{\gamma-1} \mathcal{N}^V(-\lambda) d\lambda 
= \sum_{j=1}^{N-1} \vert \lambda_j \vert^\gamma + (N-1) \vert \lambda_N \vert^\gamma,
\end{equation*}
and the remainder term vanishes if $\sum_{j} \vert \lambda_j \vert^\gamma<\infty$, concluding the proof.
		\end{proof}
With this, it now a simple calculation to derive two-sided Lieb-Thirring bounds from the estimates~\eqref{CLR}.	
		\begin{proof}[Proof of Corollary \ref{cor:LiebThirring}]
Lemma \ref{lm: sum to integral} and the upper bound \eqref{CLR} yield immediately that
		\begin{multline}
			\mathrm{tr}((-\Delta+V)_-^\gamma)
			= \gamma \int_0^{ \vert E_0 \vert} \lambda^{\gamma-1} \mathcal{N}^{V+\delta+\vert E_0 \vert}(-\lambda+\delta+\vert E_0 \vert) d\lambda \\
			\leq C_{0,d}^\frac{d}{2}\gamma \int_0^{\vert E_0 \vert} \lambda^{\gamma-1} (-\lambda+\delta+\vert E_0\vert)^\frac{d}{2} \left\vert \left\{ x\in \mathbb{R}^d:\frac{1}{u_M(x)} \leq C_{0,d}(-\lambda+\delta + \vert E_0\vert) \right\} \right\vert d\lambda \label{LT initial upper bound}.
		\end{multline}
Now since $\gamma\geq 1$, we use that $\lambda^{\gamma-1}\leq (\lambda+\delta)^{\gamma-1}$. We split the integral in two. On $[0,\frac{\vert E_0\vert}{2}]$, we use that $\vert E_0\vert-\lambda+\delta\leq \vert E_0\vert +\delta \leq \frac{\vert E_0\vert +\delta}{\delta}(\lambda + \delta)$, while $\vert E_0\vert-\lambda+\delta\leq \frac{\vert E_0\vert}{2} + \delta \leq \lambda +\delta$ holds on $[\frac{\vert E_0\vert}{2},\vert E_0\vert]$. Writing the condition $\frac{1}{u_M(x)} \leq C_{0,d}(-\lambda+\delta + \vert E_0\vert)$ as $\lambda + \delta \leq \vert E_0\vert + 2\delta - \frac{1}{C_{0,d}u_M(x)}$ and changing variables to $s=\lambda + \delta$ yields that if $\gamma\geq 1$, then
\begin{align*}
\mathrm{tr}((-\Delta+V)_-^\gamma)
&\leq C_{0,d}^\frac{d}{2}\gamma \left(1+\frac{\vert E_0\vert}{\delta}\right)^\frac{d}{2} \int_\delta^{\vert E_0 \vert +\delta}s^{\gamma-1+\frac{d}{2}}
\left\vert\left\{ x\in \mathbb{R}^d: \vert E_0\vert +2\delta - \frac{1}{C_{0,d}u_M(x)} \geq s\right\}\right\vert \\
&\leq C_{0,d}^\frac{d}{2}\frac{\gamma}{\gamma + \frac{d}{2}}\left(1+\frac{\vert E_0\vert}{\delta}\right)^\frac{d}{2}\int_{\bbR^d}\left( \vert E_0\vert +2\delta - \frac{1}{C_{0,d}u_M(x)} \right)_+^{\frac{d}{2}+\gamma} dx
\end{align*}
by the layer cake representation.

Similarly, for the lower bound, we start as in~(\ref{LT initial upper bound}) but with $c_{0,d}$ instead of $C_{0,d}$. Here, $\vert E_0\vert-\lambda+\delta\geq \frac{\vert E_0\vert}{2}+\delta > \lambda$ on $[0,\frac{\vert E_0\vert}{2}]$ while $\vert E_0\vert-\lambda+\delta\geq \delta \geq \frac{\delta}{\vert E_0\vert}\lambda$ on $[\frac{\vert E_0\vert}{2},\vert E_0\vert]$. We conclude that for $\gamma > 0$, 
		\begin{align*}
			\mathrm{tr}((-\Delta +V)_-^\gamma)
			&\geq \min\left\{1, \frac{\delta}{\vert E_0\vert}\right\}^\frac{d}{2} c_{0,d}^{d/2} \gamma \int_0^{E_0} \lambda^{\frac{d}{2}+\gamma-1} \left\vert \left\{x\in \mathbb{R}^d: \lambda \leq \vert E_0\vert+\delta - \frac{1}{c_{0,d}u_M(x)}\right\} \right\vert d\lambda \\
			&= \min\left\{1, \frac{\delta}{\vert E_0\vert}\right\}^\frac{d}{2} c^{d/2} \frac{\gamma}{\gamma + \frac{d}{2}} \int_{\mathbb{R}^d} \left( \vert E_0\vert +\delta - \frac{1}{c_{0,d}u_M(x)} \right)_-^{\frac{d}{2}+\gamma} dx.
		\end{align*}
		For the last equality we used the bound \eqref{CLR} to conclude that the measure vanishes for all $\lambda>\vert E_0\vert$ since $\mathcal{N}^V(-\lambda)=0$ for those $\lambda$.
	\end{proof}

\begin{proof}[Proof of Corollary \ref{cor:Kato LiebThirring}]	
We follow the same strategy as above, using~(\ref{Lower improved Kato},\ref{CLR improved}) instead of~(\ref{CLR}). For the lower bound, we immediately have
\begin{align*}
\mathrm{tr}((-\Delta+V)_-^\gamma) 
&=\gamma\int_0^{\vert E_0\vert}\lambda^{\gamma-1}\caN^{V}(-\lambda)d\lambda \\
&\geq C_c^\frac{d}{2}\gamma \int_0^{\vert E_0\vert}\lambda^{\gamma-1+\frac{d}{2}}\left\vert\left\{x\in\bbR^d:\lambda\leq \frac{\widetilde C_H}{c}\left(M-\frac{1}{u_M(x)}\right)\right\}\right\vert\\
&= C_c^\frac{d}{2}\frac{\gamma}{\gamma + \frac{d}{2}}\bigg(\frac{\widetilde C_H}{c}\bigg)^{\gamma + \frac{d}{2}}\int_{\bbR^d}\left(\frac{1}{u_M(x)}-M\right)_-^{\gamma + \frac{d}{2}}dx.
\end{align*}
For the upper bound, we follow the proof of Theorem~\ref{thm:Kato}(iii) to get
\begin{equation*}
\caN^V(\mu)\leq \caN^{-A_M^2(1/u_M-M)_-}\left( \mu+\varepsilon \right).
\end{equation*}
With this,
\begin{equation*}
\mathrm{tr}((-\Delta+V)_-^\gamma) 
\leq \gamma\int_0^{\vert E_0\vert}\lambda^{\gamma-1}\caN^{-A_M^2(1/u_M-M)_-}\left( -\lambda+\varepsilon \right)d\lambda
\end{equation*}
Since the counting function is c\`adl\`ag by definition~(\ref{def:counting}), we let $\varepsilon\to0^+$ and conclude that
\begin{equation*}
\mathrm{tr}((-\Delta+V)_-^\gamma) \leq \mathrm{tr}((-\Delta-A_M^2(1/u_M-M)_-)_-^\gamma)
\end{equation*}
and we use the standard Lieb-Thirring inequality to conclude. 
\end{proof}

\paragraph{Data Availability Statement.}
Data sharing is not applicable to this article as no new data were created or analyzed in this study.

\paragraph{Acknowledgements.} The three authors were supported by NSERC of Canada. S.B.\ would like to thank S.~Mayboroda and D.~Arnold for inspiring discussions.

		\bibliographystyle{plain}
		\bibliography{biblio}

	\end{document}